\documentclass[12pt, draftclsnofoot, onecolumn]{IEEEtran}
\usepackage[super]{nth}
\usepackage{amsthm}
\usepackage{algorithm}
\usepackage{amssymb}
\usepackage{multirow}

\usepackage{tabularx}  
\usepackage{arydshln,leftidx,mathtools} 

\theoremstyle{definition}
\newtheorem{theorem}{Theorem}

\newtheorem{lemma}{Lemma}

\usepackage{cite}
\ifCLASSINFOpdf
 \usepackage[pdftex]{graphicx}
\else
\fi

\usepackage{amsmath}
\usepackage{array}
\usepackage[tight,footnotesize]{subfigure}
\usepackage{courier}
\usepackage[usenames,dvipsnames]{color} 
\usepackage{color,url,xspace}
\usepackage{amssymb,amsmath}
\usepackage{graphicx}
\usepackage{colordvi}
\usepackage{epsfig}
\usepackage{endnotes}
\usepackage{verbatim}
\usepackage{textcomp}
\usepackage{subfig}
\usepackage{setspace}
\usepackage{xcolor}
\usepackage{cite}
\usepackage{times}
\usepackage{booktabs, multirow}
\usepackage{amsthm}
\usepackage{nccmath}
\usepackage{algorithmicx,algpseudocode}
\usepackage{algpseudocode,algorithm,algorithmicx}
\usepackage{caption}
\usepackage{mathtools}

\algrenewcommand\algorithmicrequire{\textbf{Precondition:}}
\algrenewcommand\algorithmicensure{\textbf{Postcondition:}}

\begin{document}
	
\title{Bi-Directional Cooperative NOMA \\without Full CSIT}

\author{
Minseok Choi, \IEEEmembership{Student Member,~IEEE,} 
Dong-Jun Han, \IEEEmembership{Student Member,~IEEE,}
Jaekyun Moon, \IEEEmembership{Fellow,~IEEE,}
}

\markboth{TRANSACTIONS ON WIRELESS COMMUNICATIONS SUBMISSION}%
{}

\maketitle

\begin{abstract}

In this paper, we propose bi-directional cooperative non-orthogonal multiple access (NOMA).
Compared to conventional NOMA, the main contributions of bi-directional cooperative NOMA can be explained in two directions:
1) The proposed NOMA system is still efficient when the channel gains of scheduled users are almost the same. 
2) The proposed NOMA system operates well without accurate channel state information (CSI) at the base station (BS). 
In a two-user scenario, the closed-form ergodic capacity of bi-directional cooperative NOMA is derived and it is proven to be better than those of other techniques. 
Based on the ergodic capacity, the algorithms to find optimal power allocations maximizing user fairness and sum-rate are presented.
Outage probability is also derived, and we show that bi-directional cooperative NOMA achieves a power gain over uni-directional cooperative NOMA and a diversity gain over non-cooperative NOMA and orthogonal multiple access (OMA).
We finally extend the bi-directional cooperative NOMA to a multi-user model.
The analysis of ergodic capacity and outage probability in two-user scenario is numerically verified. 
Also, simulation results show that bi-directional cooperative NOMA provdes better data rates than the existing NOMA schemes as well as OMA in multi-user scenario. 

\end{abstract}


\IEEEpeerreviewmaketitle

\section{Introduction}
\label{sec:Introduction}

Non-orthogonal multiple access (NOMA) based on power multiplexing has been introduced to utilize radio resources efficiently for a massive number of user terminals \cite{NOMA_basic:Docomo-Saito}.
The 4G networks mainly operate based on orthogonal multiple access (OMA), allocating orthogonal resources to multiple users.
However, as a massive number of various devices is deployed in the network, OMA is no longer able to maximize resource efficiency and to serve all devices simultaneously.
For 5G communication systems, much higher data rates are expected compared to 4G, and efficient and flexible uses of energy and spectrum have become critical issues \cite{5G:Andrews,5G:Wunder}. 
To this end, NOMA has been actively researched as a promising technology 
in 5G networks \cite{NOMA_5G:Dai,NOMA_5G:Islam}.

NOMA superposes the multi-user signals within the same frequency, time or spatial domain. 
The advanced receivers with successive interference cancellation (SIC) are typically considered to detect non-orthogonally multiplexed signals. 
In theory, NOMA provides a significant benefit in improving the cell throughput by using perfect SIC \cite{Book:Tse}. 
Performance analysis has also been conducted to examine the effectiveness of NOMA in practical environments \cite{NOMA:PerformanceAnalysis:Ding,NOMA:PerformanceAnalysis:Saito}.
User scheduling for non-orthogonally multiplexed signaling has been studied in \cite{NOMA:UserPairing:TVT-Ding}.
Optimal power allocation at the BS for NOMA users is an important issue \cite{NOMA:PowerAllocation:TWC-Choi}, for several system goals, e.g., sum-rate maximization \cite{NOMA:PowerAllocation:TWC-Yang} and user fairness \cite{NOMA:UserFairness:SPL-Timotheou}.
Recently, joint optimization of power allocation and user scheduling has been also studied for NOMA systems \cite{NOMA:URLLC:ArXiv-Choi}.

NOMA has been extensively researched in conjunction with various technologies.
There have been some studies on the system applying NOMA to MIMO \cite{MIMO-NOMA:TWC-Ding,MIMO-NOMA:TWC-Ding2} and on analyzing  ergodic capacity of MIMO-NOMA system \cite{MIMO-NOMA:WCommLetter-Sun}.
NOMA has been also considered to increase the data rate of the cell-edge user in coordinated multipoint (CoMP) \cite{NOMA-CoMP:CommLetter-Choi} and to maximize user fairness and sum-rate in distributed antenna systems \cite{NOMA-CoMP:ICC-Han}.
Recently, the application of NOMA to simultaneous wireless information and power transfer (SWIPT) \cite{NOMA-SWIPT:TWC-PD,NOMA:uni-cN-SWIPT:Liu} and physical security \cite{NOMA:security:CommLetter-Zhang} have been studied.


This paper proposes bi-directional cooperative NOMA which targets two practical channel environments: 1) when the BS knows statistical CSI but channel gain differences among users are not large, and 2) when the BS does not know CSI at all.
The proposed NOMA scheme is a kind of cooperative NOMA \cite{NOMA:uni-cN:Ding}, which allows cooperation among users via short-range communications, based on instantaneous CSI at transmitter (CSIT). 
The cooperative NOMA system of \cite{NOMA:uni-cN:Ding} improves the outage probability performance compared to conventional NOMA and it is applicable to relay communication \cite{NOMA:relay-uni-cN:Ding} and SWIPT \cite{NOMA:uni-cN-SWIPT:Liu}. 
However, it is difficult to figure out which user has the better instantaneous channel gain and which user transmits the cooperation signal to others, when only statistical CSIT or no CSIT is available.
In bi-directional cooperative NOMA system, the direction of cooperation among users can be figured out by allowing users to exchange channel information via short-range communications.

The main contributions of this paper are shown below:
\begin{itemize}
	\item 
	The closed-form ergodic capacity of bi-directional cooperative NOMA is derived, especially for a two-user scenario.
	Also, the ergodic capacity of bi-directional cooperative NOMA is shown to be better than those of other existing NOMA schemes and OMA, even when the channel variances between the users are small.
	
	\item Based on the ergodic capacity analysis, this paper presents the optimal power allocation algorithms to maximize user fairness and sum-rate for bi-directional cooperative NOMA.
	
	\item Outage probabilities of bi-directional cooperative NOMA are derived in a two-user scenario.
	For the non-SIC user, it is shown that the proposed system has a power gain over the existing NOMA schemes.
	For the SIC user, our scheme is shown to have a diversity gain over conventional NOMA and OMA.
	
	\item The extension of bi-directional cooperative NOMA to the multi-user model is presented by performing cooperation on signal-by-signal basis, not on user-by-user \cite{NOMA:uni-cN:Ding}.
	The cooperation on signal-by-signal basis does not require additional power allocations for the cooperation phases.
	
	\item Simulation results verify the analysis of ergodic capacity and outage probability. 
	Moreover, bi-directional cooperative NOMA is shown to provide better data rates than other NOMA schemes and OMA even without enough CSI, necessarily required for conventional NOMA.
	
\end{itemize}

We first propose the two-user scenario of bi-directional cooperative NOMA in Section \ref{sec:system_model}.
Ergodic capacity analysis is performed in Section \ref{sec:ergodic_capacity_analysis} and the optimal power allocation algorithms to maximize user fairness and sum-rate maximization are presented in Section \ref{sec:power_allocation}.
In Section \ref{sec:outage_prob}, outage probability of the proposed system is analyzed.
Bi-directional cooperative NOMA is extended to the multi-user model in Section \ref{sec:multi-user}, and simulation results are shown in Section \ref{sec:simulation}.
Lastly, we conclude the paper in Section \ref{sec:conclusion}.

\section{System Model}
\label{sec:system_model}

\subsection{Channel Model}
\label{subsec:channel_model}

Consider cellular downlink communications in which a BS transmits signals to two users simultaneously. 
Extension to the multi-user model will be shown in Section \ref{sec:multi-user}.
The Rayleigh fading channel from the BS to user $i$ is defined as $h_i = \sqrt{L_i}g_i$ for $i=1,2$.
$L_{i}=1/d_i^2$ denotes the slow fading, where $d_i$ is the distance from the BS to user $i$ and 
$g_i$ is a fast fading component with a complex Gaussian distribution, $g_i \sim CN(0,1)$. 

In this paper, two cases are considered in terms of the CSI knowledge: 
only statistical CSIT and no CSIT.
Here, the statistical CSIT means that the BS knows only users' channel variances. 
The BS usually allocates more power to the user having the smaller channel variance. 
The user having a larger channel variance performs SIC first to subtract inter-user interference, 
and then decodes its data. 
However, if the distances from the BS to two users are similar, then randomly generated channels cannot guarantee that the user of the larger variance experiences the stronger channel gain.
In this case, the performance gain of NOMA over OMA mostly vanishes.

In practice, there exists harsh environments where CSI is hardly known at the BS. 
For example, in an IoT environment, a clumsy device acts as a transmitter but cannot handle substantial processing tasks, i.e., channel tracking and elaborate user scheduling, so CSI is not available at the transmitter side.
In this no-CSIT case, the BS cannot judge which user has a larger or smaller channel variance, so the system should arbitrarily decide the given user to perform SIC or not. 
Also, optimal power allocation cannot be found without any CSI, so fixed power ratios for the users will be assumed.
The problem for the no-CSIT case occurs when the user with the weaker channel gain is selected to perform SIC.
In this case, there is no merit of employing conventional NOMA.
This paper proposes a new cooperative NOMA system for reliable downlink transmission for both statistical-CSIT only and no-CSIT cases.

\begin{figure} [h!]
	\centering
	\includegraphics[width=0.48\textwidth]{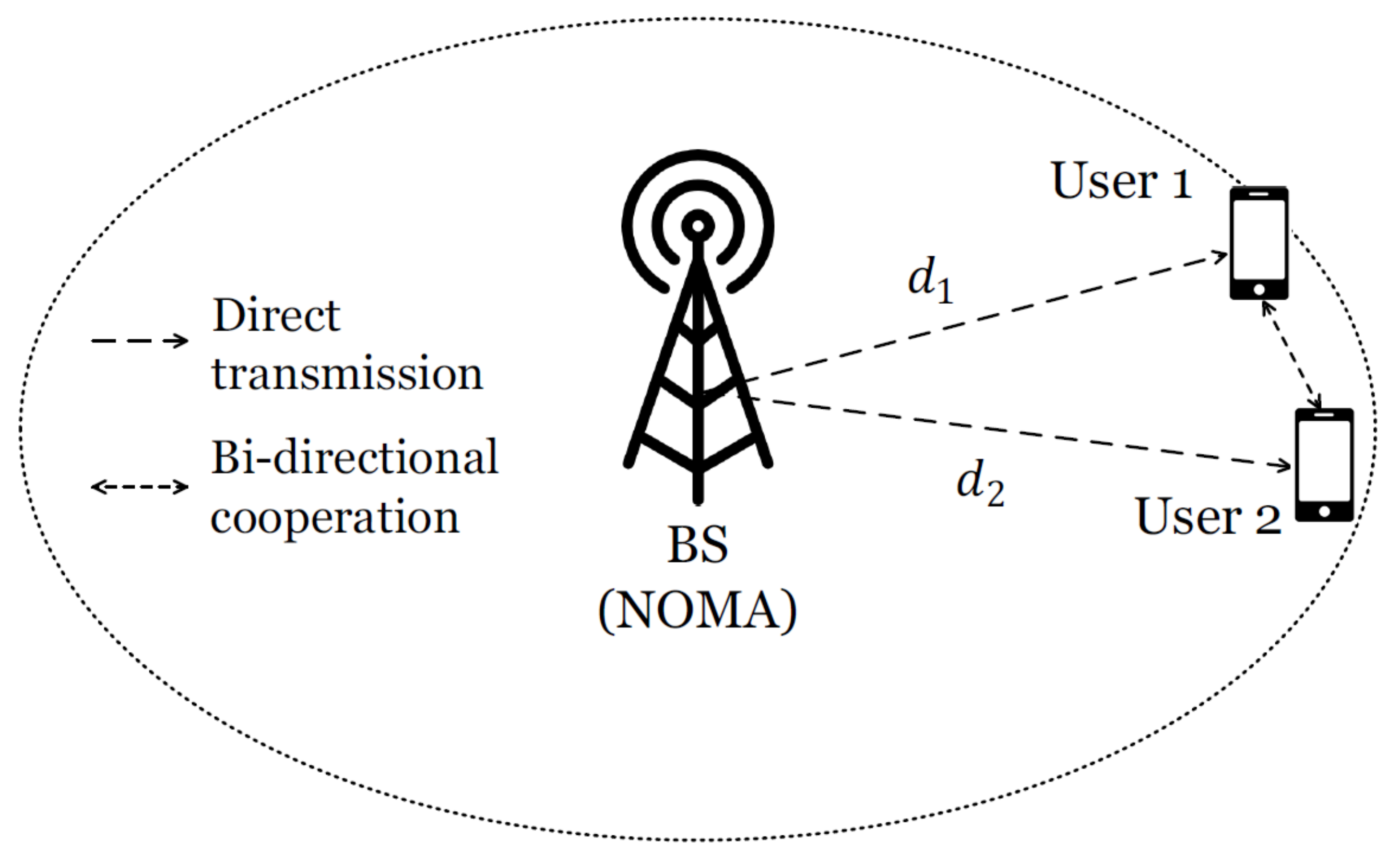}
	\caption{Bi-directional NOMA model}
	\label{fig:bi-cN_model}
\end{figure}

\subsection{Direct Transmission Phase}

Denote $\gamma_i$ as the power ratio allocated to user $i\in \{1,2\}$, satisfying $\gamma_1+\gamma_2 = 1$.
The received signal of user $i$ is given by
\begin{equation}
	r_i = h_i(\sqrt{\gamma_1}s_1 + \sqrt{\gamma_2}s_2) + n_i, \label{eq:received:user}
\end{equation}
where $s_i$, and $n_i$ are transmitted symbol and noise at user $i$, respectively, and $n_i \sim CN(0,\sigma_n^2)$. $\sigma_n^2$ is the normalized noise variance.
Assume a normalized unit power at the BS, $\mathbb{E}[|s_i|^2]=1$.
Throughout the paper, users 1 and 2 are the non-SIC user and the SIC user, respectively.

Let $V_{i,k}$ be the SINR of user $i$ to decode $s_k$. 
Then, the received SINRs at both users become
\begin{eqnarray}
V_{1,1} &=& \frac{|h_1|^2 \gamma_1}{|h_1|^2 \gamma_2 +\sigma_n^2} \\
V_{2,2} &=& \frac{|h_2|^2 \gamma_2}{\sigma_n^2}, 
\end{eqnarray}
and SINR for SIC at user 2 is given by
\begin{equation}
V_{2,1} = \frac{|h_2|^2 \gamma_1}{|h_2|^2 \gamma_2 + \sigma_n^2}.
\end{equation}
Here, user 2 performs SIC for $s_1$ first with $V_{2,1}$, and decodes $s_2$ with $V_{2,2}$.

\subsection{Channel Information Exchange Phase}
\label{subsec:channel_exchange_phase}

Since both users decode $s_1$ in the direct transmission phase, cooperation between users for improved decoding of $s_1$ is possible.
Although both users' decoding processes at the direct transmission phase can be reliable, the risk is that all users receive and exploit the cooperation signal, when decoding of the user with the weaker channel gain fails.
Therefore, the system allows only the user $i_0$, satisfying $i_0 = \underset{i\in \{1,2\}}{\arg \max} |h_i|^2$, to transmit the cooperation signal.
This indicates that transmission of the cooperation signal can be bi-directional, but actual cooperation at each time is performed at only the user of the weaker channel gain. 
To find user $i_0$, users exchange their CSI or just the received channel power in this phase.
We assume that all users are located nearby and the exchange of CSI is performed via short-range communications, so this phase would not take too much time. 
A highly crowded stadium is one example, where the distances between the BS and users do not differ greatly so the BS with statistical CSIT only and no CSIT would hardly determine the direction of cooperation appropriately.

\textit{Remark}: 
Obviously, CSI exchange among users allow the BS not to collect all users' CSI at the expense of the additional delay and signaling overhead.
Therefore, the proposed scheme can be more advantageous than existing cooperative NOMA with full CSIT, only when the CSI exchange step requires less time and overhead than transmission of CSI feedbacks.
In the case of statistical CSIT only or no CSIT, CSI feedbacks are not required, but much better data rates can be obtained by allowing the exchanges of CSI among users at the expense of the additional delay and overhead, as shwon in Section \ref{sec:simulation}. 

One more thing to remark is that even if the CSI exchange step incurs longer delays than transmission of CSI feedbacks, the proposed technique can be still beneficial over conventional schemes, especially when the channel coherence time is very short. 
In the proposed scheme, since all users already received NOMA signals from the BS and obtained the desired CSI at the direct transmission phase, even though channel conditions change during the phase of CSI exchange, users can find the appropriate direction of cooperation.
On the other hand, in conventional NOMA where the BS should collect the users' CSI feedbacks, channel conditions at the time when users send CSI feedbacks to the BS, could be changed when the BS transmits the NOMA signal to all users.
Also, as responsibility for determining the direction of cooperation signals is shifted to user sides, the BS does not need to handle substantial processing tasks for estimating the exact CSI, e.g., channel tracking. 

\subsection{Bi-Directional Cooperative Phase}
\label{subsec:bi-directional_cooperative_phase}

In this phase, the cooperation signal is transmitted from the user with stronger channel gain to the user with weaker gain.
The cooperation signal can help user 1 to decode its data, or user 2 to perform SIC better.
The received cooperation signal at user $i$ is given by
\begin{equation}
	c_{i} = g_{k,i} s_c + n_{c,i}
\end{equation}
where $i \neq k$, $g_{k,i}$ is a Rayleigh fading channel coefficient from user $k$ to user $i$, and $s_c=s_1$ here.
As mentioned in the channel information exchange phase, when $|h_1|^2 > |h_2|^2$, only $c_2$ exists, and when $|h_1|^2 < |h_2|^2$, only $c_1$ is transmitted from user 2.
The received SINR at user $i$ is
\begin{equation}
	W_{i} = \frac{|g_{k,i}|^2}{\sigma_n^2},
\end{equation}
and SINR for decoding $s_1$ is given by
\begin{equation}
Z_{\text{cN},1}^{\text{bi}} = 
\begin{cases}
\min \{ \max\{V_{1,1}, W_{1}\}, V_{2,1} \} & \text{if }|h_1|^2 < |h_2|^2 \\
\min \{ V_{1,1}, \max\{V_{2,1}, W_{2}\} \} & \text{otherwise}
\end{cases}.
\label{eq:bi-cN:R1}
\end{equation}
Here, if a certain user receives the cooperation signal, she chooses the better of the signals received in the direct transmission and bi-directional cooperative phases.
When maximal-ratio combining is exploited \cite{NOMA:uni-cN:Ding,NOMA:relay-uni-cN:Ding}, we achieve
\begin{equation}
Z_{\text{cN},1}^{\text{bi}} = 
\begin{cases}
\min \{ V_{1,1} + W_{1}, V_{2,1} \} & \text{if }|h_1|^2 < |h_2|^2 \\
\min \{ V_{1,1}, V_{2,1}+ W_{2} \} & \text{otherwise}
\end{cases}.
\label{eq:bi-cN:R1-MRC}
\end{equation}
Since only $s_1$ is shared for cooperation, SINR for decoding $s_2$ is $Z_{\text{cN},2}^{\text{bi}} = V_{2,2}$.

Assume that both users are located close to each other, i.e., $W_{1}, W_{2} \gg V_{1,1}, V_{2,1}, V_{2,2}$.
Then, $Z_{\text{cN},1}^{\text{bi}} \approx V_{2,1}$ when $|h_1|^2 < |h_2|^2$ or $Z_{\text{cN},1}^{\text{bi}} \approx V_{1,1}$ otherwise, and both (\ref{eq:bi-cN:R1}) and (\ref{eq:bi-cN:R1-MRC}) are simplified to  
\begin{equation}
\tilde{Z}_{\text{cN},1}^{\text{bi}} \simeq \max \{V_{1,1}, V_{2,1}\}.
\label{eq:bi-cN:R1-near-user}
\end{equation}
This assumption is used throughout the paper. 
The data rate of $s_i$ in bi-directional cooperative NOMA becomes $R_{\text{cN},i}^{\text{bi}} = \log_2(1+Z_{\text{cN},i}^{\text{bi}})$.

The difference of our work from \cite{NOMA:uni-cN:Ding} is that the direction of cooperation is determined at user sides by exchanging CSI among users, especially when the BS knows only statistical CSIT or no CSIT at all. 
Cooperative NOMA in \cite{NOMA:uni-cN:Ding} is based on instantaneous CSIT, and the BS can determine the user with stronger channel gain as the SIC user and the other one with weaker gain as the non-SIC user.
This makes the direction of cooperation to be always from the SIC user to the non-SIC user in the cooperative NOMA scheme of \cite{NOMA:uni-cN:Ding}.
In the statistical-CSIT only or no CSIT cases, however, there is no guarantee that the SIC user's instantaneous channel is better than that of the non-SIC user.
Therefore, exchanging the channel information among users is necessary for bi-directional cooperative NOMA to force the user with stronger channel gain transmit the cooperation signal.
We consider for comparison purposes uni-directional cooperative NOMA where direction of cooperation is always from the SIC user to the non-SIC user.

Since uni-directional cooperative NOMA only allows the SIC user to transmit the cooperation signal to the non-SIC user, SINR for decoding $s_1$ is $Z_{\text{cN},1}^{\text{uni}} = \min\{ \max\{V_{1,1}, W_{1} \}, V_{2,1} \}$, and $\tilde{Z}_{\text{cN},1}^{\text{uni}} \simeq V_{2,1}$ is obtained with the assumption that both users are located close to each other.
Also, SINR for decoding $s_1$ in conventional NOMA is $Z_{1}^{\text{N}} = \min \{V_{1,1}, V_{2,1}\}$ \cite{NOMA:PerformanceAnalysis:Ding}. 
Comparing $\tilde{Z}_{\text{cN},1}^{\text{bi}}$ with $\tilde{Z}_{\text{cN},1}^{\text{uni}}$ and $Z_1^{\text{N}}$, we can find that bi-directional cooperative NOMA exploits channel diversity. 
It is clear that $\tilde{R}_{\text{cN},1}^{\text{bi}}$ is better than or equal to $\tilde{R}_{\text{cN},1}^{\text{uni}}$ and $R_{\text{N},1}$. 
On the other hand, since the cooperation is only helpful for $s_1$, the data rate of $s_2$ is the same for all considered schemes, i.e., $R_{\text{cN},2}^{\text{bi}} = R_{\text{cN},2}^{\text{uni}} = R_{\text{N},2}$. The sum rate is obtained by $R_{\text{cN}}^{\text{bi}} = R_{\text{cN},1}^{\text{bi}} + R_{\text{cN},2}^{\text{bi}}$.

When targeted data rates are already determined, the outage event of a certain user is the criterion for determining whether the other user should receive the cooperation signal or not. 
Let $\epsilon_{1}$ and $\epsilon_{2}$ be SINR thresholds for decoding $s_1$ and $s_2$, respectively. 
Then, in this two-user scenario, even though $|h_1|^2 < |h_2|^2$, user 1 cannot receive the cooperation signal from user 2 when $V_{2,1} < \epsilon_{1}$.
Since $|h_1|^2 < |h_2|^2$, it is also clear $V_{1,1} < \epsilon_{1}$, so decoding of $s_1$ fails at both user sides.
On the other hand, when $V_{2,1} > \epsilon_{1}$, user 1 can decode $s_1$ by using cooperation from user 2, even if $V_{1,1} < \epsilon_{1}$.
The analysis of outage probability is given in Section \ref{sec:outage_prob}.

\section{Ergodic Capacity Analysis}
\label{sec:ergodic_capacity_analysis}

When users' data rates are opportunistically determined by their Quality of Service (QoS) requirements, ergodic capacity analysis is important.
Some key lemmas are established first in deriving the closed-form ergodic capacity of bi-directional cooperative NOMA.

\begin{lemma}
	For real constants $a, b>0$ and a chi-square random variable $X$, an expected value of the function $\log_2(1+\frac{aX}{b})$ becomes
	\begin{equation}
	\mathbb{E}\bigg[\log_2 \Big(1+ \frac{aX}{b} \Big)\bigg] = C_1 \Big( \frac{2a}{b} \Big),
	\end{equation}
	where $C_1(x) = \frac{1}{\ln 2} e^{1/x} \int_{1}^{\infty} \frac{1}{t} e^{-t/x} \mathrm{d}t$, for $x>0$
	\label{lemma:expected_log_chi}
\end{lemma}
\begin{proof}
	Let the nonnegative random variable $Z = \log_2(1+aX/b)$; then $\mathbb{E}[Z] = \int_{0}^{\infty} P[Z \geq z] \mathrm{d}z$ is satisfied.
	Therefore,
	\begin{equation}
	\mathbb{E}[Z] = \int_{0}^{\infty} \bigg( 1-P[Z \leq z] \mathrm{d}z \bigg) = \int_{1}^{\infty} \frac{1}{t \ln 2} e^{-\frac{b}{2a}(t-1)} \mathrm{d}t = C_1 \Big(\frac{2a}{b}\Big),
	\end{equation}
	where $t=2^z$.
\end{proof}

\begin{lemma}
	For real constants $a>0$, $b$, and a chi-square random variable $X$, 
	\begin{equation}
	\int_{0}^{\infty} e^{-bx} \log_2(1+ax) \mathrm{d}x 
	= \frac{1}{b} C_1 \Big(\frac{a}{b}\Big)
	\end{equation}
	\label{lemma:integral_chi}
\end{lemma}
\begin{proof}
	\begin{eqnarray}
	\int_{0}^{\infty} e^{-bx} \log_2(1+ax) \mathrm{d}x &=& \Big[-\frac{1}{b}e^{-bx}\log_2(1+ax)\Big]_0^{\infty} + \int_{0}^{\infty} \frac{1}{\ln 2}\cdot \frac{a}{b(1+ax)} e^{-bx} \mathrm{d}x \\
	&=& \int_{1}^{\infty} \frac{1}{b \ln 2} \cdot \frac{1}{t} e^{-\frac{b}{a}(t-1)} \mathrm{d}t = \frac{1}{b} C_1 \Big(\frac{a}{b}\Big),
	\end{eqnarray}
	where $t=1+ax$.
\end{proof}

\begin{lemma}
	$C_1 (x)$ is an increasing function of $x>0$.
\end{lemma}
\begin{proof}
	\begin{equation}
	\frac{\mathrm{d}}{\mathrm{d}x} C_1 (x) = -\frac{1}{x^2 \ln 2} e^{1/x} \int_1^{\infty} \frac{e^{-t/x}}{t} \mathrm{d}t + \frac{1}{x \ln 2} > -\frac{1}{x^2 \ln 2} \ln (1+x) + \frac{1}{x \ln 2},
	\end{equation}
	where the last inequality is satisfied according to $e^{-1/x}\ln (1+x) > \int_1^{\infty} \frac{e^{-t/x}}{t} \mathrm{d}t$ \cite{EiFunction}.
	Since $x>0$ and $x>\ln (1+x)$, $\frac{\mathrm{d}}{\mathrm{d}x} C_1 (x) > 0$, $C_1 (x)$ is an increasing function of $x>0$.
\end{proof}

With the assumption that both users are located nearby, Theorem \ref{thm:closed_ergodic_bi-cN} gives the closed-form ergodic capacity of the bi-directional cooperative NOMA system.
Also, Theorem \ref{thm:ergodic_comparison_NOMA} shows that the ergodic capacity of the bi-directional cooperative NOMA is larger than those of uni-directional cooperative NOMA and conventional NOMA, no matter which user's channel gain is larger. 

\begin{theorem}
	Assuming that both users are located close to each other, the closed-form ergodic capacity of two-user bi-directional cooperative NOMA is 
	\begin{equation}
	\mathbb{E}[\tilde{R}_{\text{cN}}^{\text{bi}}] = C_1 \Big(\frac{L_1}{\sigma_n^2}\Big) - C_1 \Big(\frac{\gamma_2 L_1}{\sigma_n^2}\Big) + C_1 \Big(\frac{L_2}{\sigma_n^2}\Big) - C_1 \Big(\frac{L_1 L_2}{(L_1+L_2)\sigma_n^2}\Big) + C_1 \Big(\frac{\gamma_2 L_1 L_2}{(L_1 + L_2)\sigma_n^2}\Big),
	\label{eq:closed_ergodic_cN_bi}
	\end{equation}
	\label{thm:closed_ergodic_bi-cN}
\end{theorem}
\begin{proof}
	See Appendix \ref{appendix:thm1}.
\end{proof}

\begin{theorem}
	
	Assuming that both users are located close to each other, 
	\begin{equation}
	\mathbb{E}[\tilde{R}^{\text{bi}}_{\text{cN}}] \geq \mathbb{E}[\tilde{R}^{\text{uni}}_{\text{cN}}] \geq \mathbb{E}[R_{\text{N}}]
	\end{equation}
	\label{thm:ergodic_comparison_NOMA}
	
\end{theorem}
\begin{proof}
	See Appendix \ref{appendix:thm2}.
\end{proof}

To verify that bi-directional cooperative NOMA is applicable, 
comparison with OMA is also necessary.
Theorem \ref{theorem:ergodic_inequality_oma} shows that the ergodic capacity of bi-directional cooperative NOMA is better than that of OMA when the channel variances of two users are identical.
Lemma \ref{lemma:Cx-Cax} is introduced first before stating Theorem \ref{theorem:ergodic_inequality_oma}.

\begin{lemma}
	$C_1(x) - C_1(\beta x)$ is an increasing function of $x>0$, for any $0<\beta<1$.
	\label{lemma:Cx-Cax}
\end{lemma}
\begin{proof}
	According to (\ref{eq:R1_uni-cN_closed}) and (\ref{eq:R2_closed}), the closed-form ergodic capacity of the uni-directional cooperative NOMA system becomes 
	\begin{equation}
	\mathbb{E}[\tilde{R}_{\text{cN}}^{\text{uni}}] = C_1 \Big(\frac{L_1}{\sigma_n^2}\Big) - C_1 \Big(\frac{\gamma_2 L_1}{\sigma_n^2}\Big) + C_1 \Big(\frac{\gamma_2 L_2}{\sigma_n^2}\Big).
	\label{eq:closed_ergodic_cN_omni}
	\end{equation}
	
	According to (\ref{eq:closed_ergodic_cN_bi}) and (\ref{eq:closed_ergodic_cN_omni}), the following inequality holds by Theorem \ref{thm:ergodic_comparison_NOMA},
	\begin{eqnarray}
	&& C_1 \Big(\frac{L_1}{\sigma_n^2}\Big) - C_1 \Big(\frac{\gamma_2 L_1}{\sigma_n^2}\Big) + C_1 \Big(\frac{L_2}{\sigma_n^2}\Big) - C_1 \Big(\frac{L_1 L_2}{(L_1+L_2)\sigma_n^2}\Big) + C_1 \Big(\frac{\gamma_2 L_1 L_2}{(L_1 + L_2)\sigma_n^2}\Big) \nonumber \\
	&&~~~~~\geq C_1 \Big(\frac{L_1}{\sigma_n^2}\Big) - C_1 \Big(\frac{\gamma_2 L_1}{\sigma_n^2}\Big) + C_1 \Big(\frac{\gamma_2 L_2}{\sigma_n^2}\Big) \\
	&\Leftrightarrow& C_1 \Big(\frac{L_2}{\sigma_n^2}\Big) - C_1 \Big(\frac{\gamma_2 L_2}{\sigma_n^2}\Big) - \bigg\{ C_1 \Big(\frac{L_1 L_2}{(L_1+L_2)\sigma_n^2}\Big) - C_1 \Big(\frac{\gamma_2 L_1 L_2}{(L_1 + L_2)\sigma_n^2}\Big) \bigg\} \geq 0.
	\label{eq:temp3}
	\end{eqnarray}
	Equation (\ref{eq:temp3}) holds for any $L_1, L_2>0$ and $0<\gamma_2<1$. 
	Thus, $C_1(x) - C_1(\beta x)$ is an increasing function of $x>0$.
\end{proof}

\begin{theorem}
	When $L_1 = L_2 = L$, $\mathbb{E}[\tilde{R}_{\text{cN}}^{\text{bi}}] > \mathbb{E}[R_{\text{O}}]$, provided $\frac{L}{\sigma_n^2}, \frac{\gamma_1 L}{\alpha_1 \sigma_n^2},\frac{\gamma_2 L}{\alpha_2 \sigma_n^2}>1$.
	\label{theorem:ergodic_inequality_oma}
\end{theorem}
\begin{proof}
	See Appendix \ref{appendix:thm3}.
\end{proof}

It is reasonable that $\mathbb{E}[\tilde{R}_{\text{cN}}^{\text{bi}}]$ becomes much larger than $\mathbb{E}[R_{\text{O}}]$ as $L_2$ increases above $L_1$.
Also, $\mathbb{E}[\tilde{R}_{\text{cN}}^{\text{bi}}] - \mathbb{E}[R_{\text{O}}] > 0$ when $L_1=L_2$ by Theorem \ref{theorem:ergodic_inequality_oma}, so it can also be expected that $\mathbb{E}[\tilde{R}_{\text{cN}}^{\text{bi}}]$ could be still larger than $\mathbb{E}[R_{\text{O}}]$ when $L_1 = L_2+\delta$ for small $\delta>0$.
In Section \ref{sec:simulation}, numerical results show that bi-directional cooperative NOMA still has a rate gain compared to OMA even when the SIC user (user 2) experiences the weaker channel than the non-SIC user (user 1).

\section{Optimal Power Allocation Rule}
\label{sec:power_allocation}

Based on ergodic capacity analysis, we present the optimal power allocation rule for bi-directional cooperative NOMA.
Two optimization goals are considered: user fairness and sum-rate.
Note that the BS should know statistical CSI at least for the optimal power allocation, and we do not consider the no-CSIT case here.
Assume $L_1 < L_2$ in this section.

\subsection{User Fairness Problem}

As in \cite{NOMA:UserFairness:SPL-Timotheou,NOMA-CoMP:ICC-Han}, the max-min optimization problem is formulated for user fairness as
\begin{equation}
\gamma_{2}^{*}=\underset{0<\gamma_2 <1}{\arg \max}~\min (\mathbb{E}[\tilde{R}_{\text{cN},1}^{\text{bi}}], \mathbb{E}[R_{\text{cN},2}^{\text{bi}}]),
\label{eq:min-max-prob}
\end{equation}
where $\gamma_2^{*}$ is the optimal power ratio for user 2 and recall that $\gamma_1 + \gamma_2 =1$. 
The following lemma helps to solve the above optimization problem.

\begin{lemma}
	$\mathbb{E}[\tilde{R}_{\text{cN},1}^{\text{bi}}]$ is a decreasing function of $\gamma_2$ and $\mathbb{E}[R_{\text{cN},2}^{\text{bi}}]$ is an increasing function of $\gamma_2$.
	\label{lemma:R_1_R_2}
\end{lemma}
\begin{proof}
	$\tilde{R}_{\text{cN},1}^{\text{bi}} = \max\{Z_1, Z_2\}$, where $Z_1 = \log_2 (1+ \frac{ |h_1|^2 \gamma_1 }{ |h_1|^2 \gamma_2  + \sigma_n^2}) = \log_2(1+\frac{|h_1|^2}{\sigma_n^2}) - \log_2(1+ \frac{|h_1|^2 \gamma_2}{\sigma_n^2})$, and $Z_2 = \log_2 ( 1+ \frac{ |h_2|^2 \gamma_1 }{ |h_2|^2 \gamma_2 + \sigma_n^2}) = \log_2(1+\frac{|h_2|^2}{\sigma_n^2}) - \log_2(1+ \frac{|h_2|^2 \gamma_2}{\sigma_n^2})$. 
	Since $\log_2(1+\frac{|h_i|^2\gamma_2}{\sigma_n^2})$ is an increasing function of $\gamma_2$,
	$Z_1$ and $Z_2$ are decreasing functions so $\mathbb{E}[\tilde{R}_{\text{cN},1}^{\text{bi}}]$ is a decreasing function of $\gamma_2$.
	Also, $R_{\text{cN},2}^{\text{bi}} = \log_2(1+\frac{|h_2|^2 \gamma_2}{\sigma_n^2})$, so $\mathbb{E}[R_{\text{cN},2}^{\text{bi}}]$ is an increasing function of $\gamma_2$.
\end{proof}

\begin{algorithm}
	\caption{Bisection method for power allocation of user fairness problem
		\label{algo:user_fairness}}
	\begin{algorithmic}[1]
		\State{Initialize $\gamma_{-}=0$, $\gamma_{+}=1$.}
		\While{$\gamma_{+} - \gamma_{-} \geq \epsilon$}{
			\State{$\gamma_2^{*} = (\gamma_{+} + \gamma_{-})/2$}
			\If{$\mathbb{E}[\tilde{R}_{\text{cN},1}^{\text{bi}}] < \mathbb{E}[R_{\text{cN},2}^{\text{bi}}]$}
			$\gamma_{+} = \gamma_{2}^{*}$
			\Else 
			~~$\gamma_{-} = \gamma_{2}^{*}$
			\EndIf
		}
		\EndWhile
	\end{algorithmic}
\end{algorithm}

By Lemma \ref{lemma:R_1_R_2}, the optimal solution of (\ref{eq:min-max-prob}) is directly obtained when $\mathbb{E}[\tilde{R}_{\text{cN},1}^{\text{bi}}] = \mathbb{E}[R_{\text{cN},2}^{\text{bi}}]$.
Since $\mathbb{E}[\tilde{R}_{\text{cN},1}^{\text{bi}}], \mathbb{E}[{R}_{\text{cN},2}^{\text{bi}}] \geq 0$ for any $\gamma_1, \gamma_2 \in [0,1]$, $\mathbb{E}[\tilde{R}_{\text{cN},1}^{\text{bi}}]=0$ at $\gamma_1=0$, and $\mathbb{E}[{R}_{\text{cN},2}^{\text{bi}}]=0$ at $\gamma_2=0$, the solution of $\mathbb{E}[\tilde{R}_{\text{cN},1}^{\text{bi}}] = \mathbb{E}[{R}_{\text{cN},2}^{\text{bi}}]$ would satisfy $0\leq \gamma_2^{*} \leq 1$.
However, the closed-form solution of $\gamma_2^{*}$ is difficult to derive because of the expectation operations.
Therefore, the bisection method is used to solve (\ref{eq:min-max-prob}).
Algorithm \ref{algo:user_fairness} shows the detail.

\subsection{Max-Sum-Rate Problem}

To maximize the sum-rate of NOMA system, allocating all power to the strong user is a simple solution.
However, it destroys user-fairness completely; the sum-rate performance is usually studied under a minimum rate constraint as in \cite{NOMA-CoMP:ICC-Han}. 
The problem can be formulated as 
\begin{align}
&\gamma_{2}^{*}=\underset{0<\gamma_2 <1}{\arg \max}~\mathbb{E}[\tilde{R}_{\text{cN},1}^{\text{bi}}] + \mathbb{E}[R_{\text{cN},2}^{\text{bi}}] \label{eq:max-sum-rate_prob}\\
&\text{s.t.} \ \text{min}(\mathbb{E}[\tilde{R}_{\text{cN},1}^{\text{bi}}], \mathbb{E}[R_{\text{cN},2}^{\text{bi}}])\geq R_{t}
\label{eq:max-sum-rate_const}
\end{align}
where $R_t$ is the minimum data rate constraint for user fairness.
Lemma \ref{lemma:R_1+R_2_NearUserAssumption} is introduced for solving the above optimization problem.

\begin{lemma}
	$\mathbb{E}[\tilde{R}_{\text{cN},1}^{\text{bi}}] + \mathbb{E}[{R}_{\text{cN},2}^{\text{bi}}]$ is a decreasing function of $\gamma_2$, for any $L_1, L_2>0$.
	\label{lemma:R_1+R_2_NearUserAssumption}
\end{lemma}
\begin{proof}
	According to (\ref{eq:closed_ergodic_cN_bi}), 
	\begin{eqnarray}
	\frac{\mathrm{d}}{\mathrm{d}\gamma_2}\big( \mathbb{E}[\tilde{R}_{\text{cN},1}^{\text{bi}}] + \mathbb{E}[{R}_{\text{cN},2}^{\text{bi}}] \big) = \frac{\mathrm{d}}{\mathrm{d}\gamma_2} \bigg\{ - C_1 \Big(\frac{\gamma_2 L_1}{\sigma_n^2}\Big) + C_1 \Big(\frac{\gamma_2 L_1 L_2}{(L_1 + L_2)\sigma_n^2}\Big) \bigg\} < 0,
	\end{eqnarray}
	by Lemma \ref{lemma:Cx-Cax} and $0<\frac{L_1}{L_1+L_2}<1$.
\end{proof}

By Lemmas \ref{lemma:R_1_R_2} and \ref{lemma:R_1+R_2_NearUserAssumption}, the optimal solution of (\ref{eq:max-sum-rate_prob}) is obtained when $\mathbb{E}[{R}_{\text{cN},2}^{\text{bi}}] = R_t$.
Similar to the user fairness problem, the bisection method can be used to find $\gamma_2^{*}$ of (\ref{eq:max-sum-rate_prob}).
However, if the assumption that both users are located nearby is not satisfied, Lemma \ref{lemma:R_1+R_2_NearUserAssumption} does not hold anymore. 
In this case, we consider some special cases depending on the relative amounts of $W_{1}$ and $W_{2}$ compared to $V_{1,1}$ and $V_{2,1}$.
Suppose $V_{1,1} < V_{2,1}$, then Lemma \ref{lemma:R_1+R_2_NearUserAssumption} holds when $V_{2,1} < W_1$. 
However, when $V_{1,1} < W_1 < V_{2,1}$, $R_{\text{cN},1}^{\text{bi}}$ becomes $\log_2(1+\frac{|g_c|^2}{\sigma_n^2})$, so $\mathbb{E}[R_{\text{cN},1}^{\text{bi}}]$ does not depend on the power allocation ratio.
Then, $\mathbb{E}[R_{\text{cN},1}^{\text{bi}}]+\mathbb{E}[R_{\text{cN},2}^{\text{bi}}]$ becomes an increasing function of $\gamma_2$ because $\mathbb{E}[R_{\text{cN},2}^{\text{bi}}]$ does.
On the other hand, when $W_1 < V_{1,1}$, $R_{\text{cN},1}^{\text{bi}}$ becomes $\log_2(1+\frac{\gamma_1|h_1|^2}{\gamma_2|h_2|^2 + \sigma_n^2})$, and it can be proven that $\mathbb{E}[R_{\text{cN},1}^{\text{bi}}]+\mathbb{E}[R_{\text{cN},2}^{\text{bi}}]$ is a decreasing function of $\gamma_2$ in a way similar to the proof of Lemma \ref{lemma:R_1+R_2_NearUserAssumption}.

The situation where $V_{1,1} > V_{2,1}$ can be also considered similar to $V_{1,1} < V_{2,1}$.
However, this case is not applied to solve the max-sum-rate problem (\ref{eq:max-sum-rate_prob}). 
The reason is that $\mathbb{E}[V_{1,1}], \mathbb{E}[V_{2,1}], \mathbb{E}[W_1]$ and $\mathbb{E}[W_2]$ are used instead of $V_{1,1}, V_{2,1}, W_1$ and $W_2$ in the statistical CSIT case.
This approximation does not consider the case of $V_{1,1} > V_{2,1}$, because we assume $L_2 > L_1$ first so $\mathbb{E}[V_{1,1}] < \mathbb{E}[V_{2,1}]$ always.
Thus, this approximation makes $\gamma_2^{*}$ of (\ref{eq:max-sum-rate_prob}) a suboptimal solution.

\begin{algorithm}
	\caption{Bisection method of power allocation for max-sum-rate problem
		\label{algo:max_sum_rate}}
	\begin{algorithmic}[1]
		\State{Initialize $\gamma_{-}=0$, $\gamma_{+}=1$.}
		\While{$\gamma_{+} - \gamma_{-} \geq \epsilon$}{
			\State{$\gamma_2^{*} = (\gamma_{+} + \gamma_{-})/2$}
			
			\If{$R_0 > R_t$}
			{
				
				\If{$\min \{ \mathbb{E}[R_{\text{cN},1}^{\text{bi}}], \mathbb{E}[R_{\text{cN},2}^{\text{bi}}] \} \leq R_t$}
				{
					\If{$\mathbb{E}[R_{\text{cN},1}^{\text{bi}}] > \mathbb{E}[R_{\text{cN},2}^{\text{bi}}]$} {$\gamma_{-} = \gamma_2^{*}$}
					\Else 
					{
						$\gamma_{+} = \gamma_2^{*}$
					}
					\EndIf
				}
				\Else
				{
					\If{$\mathbb{E}[V_{1,1}] < \mathbb{E}[W_1] < \mathbb{E}[V_{2,1}]$}
					{
						$\gamma_{-} = \gamma_2^{*}$
					}
					\Else
					{
						$\gamma_{+} = \gamma_2^{*}$
					}
					\EndIf
				}
				\EndIf
				
			}
			\Else
			{
				System outage occurs.
			}
			\EndIf
			
		}
		\EndWhile
	\end{algorithmic}
\end{algorithm}

In summary, if $V_{1,1} < W_1 < V_{2,1}$, $\mathbb{E}[R_{\text{cN},1}^{\text{bi}}] + \mathbb{E}[R_{\text{cN},2}^{\text{bi}}]$ is an increasing function of $\gamma_2$, so the solution is obtained when $\mathbb{E}[R_{\text{cN},1}^{\text{bi}}] = R_t$.
If not, $\mathbb{E}[R_{\text{cN},1}^{\text{bi}}] + \mathbb{E}[R_{\text{cN},2}^{\text{bi}}]$ is a decreasing function of $\gamma_2$, and $\gamma_2^{*}$ is found when $\mathbb{E}[R_{\text{cN},2}^{\text{bi}}] = R_t$.
Based on these behaviors, the suboptimal bisection method for maximizing the sum-rate of bi-directional cooperative NOMA in the statistical CSIT case is presented in Algorithm \ref{algo:max_sum_rate}.
Note that the outage event occurs when the minimum rate constraint (\ref{eq:max-sum-rate_const}) is not satisfied. 
In addition, according to Lemma \ref{lemma:R_1_R_2}, we can recognize that $R_0 = \mathbb{E}[R^{\text{bi}}_{\text{cN},1}] = \mathbb{E}[R^{\text{bi}}_{\text{cN},2}]$ should be larger than $R_t$; otherwise, the system cannot avoid outage.

\section{Outage Probability}
\label{sec:outage_prob}

When the targeted data rates, $R_{t,1}$ and $R_{t,2}$, are determined by the users' QoS requirements, the outage probability is an important performance criterion.
If the outage event occurs at the non-SIC user, the SIC user does not use the cooperation signal, and outage of SIC user does not allow the cooperation from the SIC user to the non-SIC user.
The outage probability at the non-SIC user (user 1) in bi-directional cooperative NOMA is given by
\begin{equation}
P_{\text{cN},1}^{\text{bi}} = P\{V_{1,1} < \epsilon_{1},~V_{2,1} < \epsilon_{1}\} + P\{ \max\{ V_{1,1},~W_1 \} < \epsilon_{1},~V_{2,1} > \epsilon_{1} \},
\label{eq:outage1_bi-cN}
\end{equation}
where $\epsilon_i = 2^{R_{t,i}}-1$.

Again, $|h_i|^2=L_i X_i/2$ for $i \in \{1,2\}$.
The first term of (\ref{eq:outage1_bi-cN}) becomes
\begin{equation}
P\bigg\{ \frac{\gamma_1 L_1 X_1}{\gamma_2 L_1 X_1 + 2\sigma_n^2} < \epsilon_1 \bigg\} \cdot P\bigg\{ \frac{\gamma_1 L_2 X_2}{\gamma_2 L_2 X_2 + 2\sigma_n^2} < \epsilon_1 \bigg\} 
=
\begin{cases}
1 & \text{if } \frac{\gamma_1}{\gamma_2}<\epsilon_1 \\
(1- e^{-\xi / L_1}) (1- e^{ - \xi / L_2 } ) & \text{otherwise} \\
\end{cases}
\label{eq:outage1_bi-cN_temp1}
\end{equation}
where $\epsilon_i = 2^{R_{ti}}-1$ and $\xi = \frac{\sigma_n^2 \epsilon_1}{\gamma_1 - \epsilon_1 \gamma_2}$, and the second term of (\ref{eq:outage1_bi-cN}) becomes
\begin{eqnarray}
&&P\bigg\{ \frac{\gamma_1 L_1 X_1}{\gamma_2 L_1 X_1 + 2\sigma_n^2} < \epsilon_1 \bigg\} \cdot P\bigg\{ \frac{L_c X_c}{2\sigma_n^2} < \epsilon_1 \bigg\} \cdot P\bigg\{ \frac{\gamma_1 L_2 X_2}{\gamma_2 L_2 X_2 + 2\sigma_n^2} > \epsilon_1 \bigg\} \\
&&~~~=
\begin{cases}
0 & \text{if } \frac{\gamma_1}{\gamma_2}<\epsilon_1 \\
( 1- e^{ - \xi/L_1} ) ( 1- e^{ -\sigma_n^2 \epsilon_1/L_c } ) e^{ -\xi / L_2} & \text{otherwise}
\end{cases}
\label{eq:outage1_bi-cN_temp2}
\end{eqnarray}

By (\ref{eq:outage1_bi-cN}), (\ref{eq:outage1_bi-cN_temp1}) and (\ref{eq:outage1_bi-cN_temp2}), $P_{\text{cN},1}^{\text{bi}}$ is given by 
\begin{equation}
P_{\text{cN},1}^{\text{bi}} =
\begin{cases}
1 & \text{if } \frac{\gamma_1}{\gamma_2}<\epsilon_1 \\
(1- e^{ - \xi / L_1} ) (1- e^{ -\sigma_n^2 \epsilon_1/L_c - \xi / L_2} ) & \text{otherwise}
\end{cases}.
\end{equation}
$P_{\text{cN},1}^{\text{bi}}$ conditioned on $\gamma_1 / \gamma_2>\epsilon_1$ is approximated in the high SNR region by
\begin{equation}
P_{\text{cN},1}^{\text{bi}} \approx 
\frac{\sigma_n^2 \epsilon_1}{L_1 (\gamma_1 - \epsilon_1 \gamma_2)} \cdot \Big( \frac{\sigma_n^2 \epsilon_1}{L_c} + \frac{\sigma_n^2 \epsilon_1}{L_2(\gamma_1-\epsilon_1\gamma_2)} \Big),
\label{eq:outage1_bi-cN_highSNR}
\end{equation}
and it indicates user 1 achieves a diversity order of 2.

Here, the outage probability of user 1 in bi-directional cooperative NOMA is the same as that of uni-directional cooperative NOMA, i.e., $P_{\text{cN},1}^{\text{bi}} = P_{\text{cN},1}^{\text{uni}}$, because user 1 receives the cooperation signal from user 2 also in uni-directional cooperative NOMA.
On the other hand, conventional NOMA and OMA are different. 
For conventional NOMA, the outage probability of user 1 of conventional NOMA, $P_{\text{N},1}$, is given by 
\begin{eqnarray}
P_{\text{N},1} = P\Big\{ \frac{\gamma_1 L_1 X_1}{\gamma_2 L_1 X_1 + 2\sigma_n^2} < \epsilon_1 \Big\} &=& 1 - e^{ -\xi / L_1} \label{eq:outage1_conv-N} \\
&\approx& \frac{\sigma_n^2 \epsilon_1}{L_1(\gamma_1 - \epsilon_1\gamma_2)},
\label{eq:outage1_conv-N_highSNR}
\end{eqnarray}
and it just has a diversity order of 1. 
Equation (\ref{eq:outage1_conv-N_highSNR}) is achieved by a high-SNR approximation.

Likewise, user 1 of OMA also achieves a diversity order of 1, as shown in (\ref{eq:outage1_OMA}) and (\ref{eq:outage1_OMA_highSNR}).
\begin{eqnarray}
P_{\text{O},1} = P\Big\{ \frac{\gamma_1L_1X_1}{2\alpha_1 \sigma_n^2} < \epsilon_{O,1} \Big\} &=& 1-\exp\Big\{ -\frac{\alpha_1 \sigma_n^2 \epsilon_{O,1}}{\gamma_1L_1} \Big\} \label{eq:outage1_OMA}\\
&\approx& \frac{\alpha_1 \sigma_n^2 \epsilon_{O,1}}{\gamma_1L_1},
\label{eq:outage1_OMA_highSNR}
\end{eqnarray}
where $\epsilon_{O,i} = 2^{R_{t,i}/\alpha_i}-1$ for $i \in \{1,2\}$.
We can easily note that bi- and uni-directional cooperative NOMA systems can achieve multiuser diversity.

Next, consider the outage probability of user 2 in bi-directional cooperative NOMA system.
If user 1 avoids the outage event, user 2 can use the cooperation signal transmitted from user 1 for its SIC process.
However, if the user 1's data rate is less than $R_{t,1}$, the cooperation from user 1 to user 2 cannot be performed. 
The outage probability of user 2 in bi-directional cooperative NOMA is given by
\begin{eqnarray}
P_{\text{cN},2}^{\text{bi}} &=& P\{ (V_{2,2} < \epsilon_{2} \cup V_{2,1} < \epsilon_{1}),~V_{1,1} < \epsilon_{1} \} \nonumber \\
&&~~+ P\{ (V_{2,2} < \epsilon_{2} \cup \max \{V_{2,1},~W_2\} < \epsilon_{1}),~V_{1,1} > \epsilon_{1} \}
\label{eq:outage2_bi-cN}
\end{eqnarray}

The first term of (\ref{eq:outage2_bi-cN}) becomes 
\begin{eqnarray}
&&(1- P\{Z_2 > \epsilon_{2},~ V_{2,1} > \epsilon_{1} \})P\{V_{1,1} < \epsilon_{1}\} \\
&&~~~= \bigg(1- P\Big\{ \frac{\gamma_2L_2X_2}{2\sigma_n^2} > \epsilon_2,~ \frac{\gamma_1 L_2 X_2}{\gamma_2 L_2 X_2 + 2\sigma_n^2} > \epsilon_1 \Big\} \bigg) P\Big\{ \frac{\gamma_1 L_1 X_1}{\gamma_2 L_1 X_1 + 2\sigma_n^2} < \epsilon_1 \Big\} 
\label{eq:outage2_bi-cN_temp1}
\end{eqnarray}

Similarly, the second term of (\ref{eq:outage2_bi-cN}) becomes 

\begin{eqnarray}
&&\big[ P\{ V_{2,2} < \epsilon_{2}\} + P\{\max\{V_{2,1},~W_2\} < \epsilon_{1} \} \nonumber \\
&&~~~- P\{ V_{2,2} < \epsilon_{2} \cap \max\{V_{2,1},~W_2\} < \epsilon_{1} \} \big] \cdot P\{V_{1,1} > \epsilon_{1}\},
\label{eq:outage1_bi-cN_2ndterm}
\end{eqnarray}
where 
\begin{eqnarray}
&&P\{\max\{V_{2,1},~W_2\} < \epsilon_{1} \} = P\{V_{2,1} < \epsilon_{1}\} P\{W_2 < \epsilon_{1}\} \\
&&~~~~~~~~= 
\begin{cases}
1 & \text{if } \frac{\gamma_1}{\gamma_2}<\epsilon_1 \\
( 1- e^{ -\xi / L_2} ) ( 1- e^{ -\sigma_n^2\epsilon_1/L_c } ) & \text{otherwise} \\
\end{cases}
\label{eq:outage2_bi-cN_temp2}
\end{eqnarray}
and
\begin{equation}
P\{ V_{2,2} < \epsilon_{2} \cap \max\{V_{2,1},~W_2\} < \epsilon_{1} \}
=
\begin{cases}
1 & \text{if } \frac{\gamma_1}{\gamma_2}<\epsilon_1 \\
( 1- e^{ -\xi / L_2} ) ( 1- e^{ -\sigma_n^2\epsilon_1/L_c } ) & \text{else if } \frac{\epsilon_2}{\gamma_2} > \frac{\epsilon_1}{\gamma_1-\epsilon_1\gamma_2} \\
( 1- e^{ -\sigma_n^2 \epsilon_2/\gamma_2 L_2 } ) ( 1- e^{ -\sigma_n^2\epsilon_1/L_c } ) & \text{else}
\end{cases}
\label{eq:outage2_bi-cN_temp3}
\end{equation}

Thus, according to (\ref{eq:outage2_bi-cN}), (\ref{eq:outage2_bi-cN_temp1}), (\ref{eq:outage2_bi-cN_temp2}), and (\ref{eq:outage2_bi-cN_temp3}), $P_{\text{cN},2}^{\text{bi}}$ is given by
\begin{equation}
P_{\text{cN},2}^{\text{bi}} = 
\begin{cases}
1 & \text{if } \frac{\gamma_1}{\gamma_2}<\epsilon_1 \\

1 - e^{ -\sigma_n^2 \epsilon_2 / \gamma_2 L_2 } \approx \frac{\sigma_n^2 \epsilon_2}{\gamma_2 L_2}
& \text{else if } \frac{\epsilon_2}{\gamma_2} > \frac{\epsilon_1}{\gamma_1-\epsilon_1\gamma_2} \\

( 1- e^{ - \xi / L_2 } ) - e^{ - \sigma_n^2 \epsilon_1 / L_c - \xi / L_1 } ( e^{ - \sigma_n^2 \epsilon_2 / \gamma_2 L_2 } - e^{ - \xi / L_2 } ) \\
\approx \frac{\sigma_n^2 \epsilon_1}{L_2(\gamma_1 - \epsilon_1 \gamma_2)} - \Big( \frac{\sigma_n^2 \epsilon_1}{L_2(\gamma_1 - \epsilon_1 \gamma_2)} - \frac{\sigma_n^2 \epsilon_2}{\gamma_2 L_2} \Big) \Big( 1 - \frac{\sigma_n^2 \epsilon_1}{L_c} - \frac{\sigma_n^2	\epsilon_1}{L_1(\gamma_1 - \epsilon_1 \gamma_2)} \Big)
& \text{else}
\end{cases}.  \label{eq:outage2_bi-cN_closed} 
\end{equation}
The approximations in (\ref{eq:outage2_bi-cN_closed}) are obtained in the high SNR region.
Unlike $P_{\text{cN},1}^{\text{bi}}$, $P_{\text{cN},2}^{\text{bi}}$ has a diversity order of 1. 
The reason is that user 2 should perform SIC before its data decoding, even though the cooperation signal from user 1 could help its SIC process. 
Uni-directional cooperative NOMA does not allow user 2 to receive the cooperation signal, so its outage probability is obtained by

\begin{eqnarray}
P_{\text{cN},2}^{\text{uni}} &=& P\{ V_{2,2} < \epsilon_{2} \cup V_{2,1} < \epsilon_{1} \} = 1 - P\{ V_{2,2} > \epsilon_{2},~ V_{2,1} > \epsilon_{1} \} \\
&=& \begin{cases}
1 & \text{if } \frac{\gamma_1}{\gamma_2}<\epsilon_1 \\
1 - \exp\{-\frac{\sigma_n^2 \epsilon_2}{\gamma_2 L_2}\} \approx \frac{\sigma_n^2 \epsilon_2}{\gamma_2 L_2} & \text{else if }\frac{\epsilon_2}{\gamma_2} > \frac{\epsilon_1}{\gamma_1-\epsilon_1 \gamma_2}\\
1 - \exp\{ -\frac{\sigma_n^2 \epsilon_1}{L_2 (\gamma_1 - \epsilon_1 \gamma_2)} \} \approx \frac{\sigma_n^2 \epsilon_1}{L_2 (\gamma_1 - \epsilon_1 \gamma_2)} & \text{else}
\end{cases}.
\label{eq:outage2_uni-cN_closed}
\end{eqnarray}
The high SNR approximation in (\ref{eq:outage2_uni-cN_closed}) shows user 2 in uni-directional cooperative NOMA realizes a diversity order of 1, the same as bi-directional cooperative NOMA.
Conditioned on $\frac{\gamma_1}{\gamma_2}>\epsilon_1$, $P_{\text{cN},2}^{\text{bi}} = P_{\text{cN},2}^{\text{uni}}$ when $\frac{\epsilon_2}{\gamma_2} > \frac{\epsilon_1}{\gamma_1-\epsilon_1 \gamma_2}$, but $P_{\text{cN},2}^{\text{bi}}$ has a power gain compared to $P_{\text{cN},2}^{\text{uni}}$ when $\frac{\epsilon_2}{\gamma_2} < \frac{\epsilon_1}{\gamma_1-\epsilon_1 \gamma_2}$. 
Since there is no cooperation signal from user 1 to user 2 in uni-directional cooperative NOMA,  $P_{\text{cN},2}^{\text{uni}}$ is the same as the outage probability of user 2 of conventional NOMA, $P_{\text{cN},2}^{\text{uni}} = P_{\text{N},2}$.
Meanwhile, the outage probability of the user 2 of OMA is given by
\begin{equation}
P_{\text{O},2} = 1 - \exp \Big\{ -\frac{\alpha_2 \sigma_n^2 \epsilon_{O,2}}{\gamma_2 L_2} \Big\} \approx \frac{\alpha_2 \sigma_n^2 \epsilon_{O,2}}{\gamma_2 L_2},
\end{equation}
and it also realizes a diversity order of 1.

In summary, bi-directional cooperative NOMA achieves a power gain for the SIC user (user 2) conditioned on $\frac{\epsilon_2}{\gamma_2} < \frac{\epsilon_1}{\gamma_1-\epsilon_1 \gamma_2}$ and $\frac{\gamma_1}{\gamma_2} > \epsilon_1$, compared to uni-directional cooperative NOMA and conventional NOMA.
On the other hand, the non-SIC user of bi- or uni-directional cooperative NOMA scheme achieves better diversity order than that of conventional NOMA and OMA, as shown by (\ref{eq:outage1_bi-cN_highSNR}).

\section{Extension to Multi-User Scenario}
\label{sec:multi-user}

Thus far, we considered the two-user model for bi-directional cooperative NOMA.
However, the proposed system can be extended to the multi-user scenario. 
Assume that the BS serves $K$ users by NOMA. 
Bi-directional cooperative NOMA consists of $K+2$ phases. 
The first and the second phases correspond to direct transmission and channel information exchange, respectively, and others are for cooperation. 
In the statistical-CSIT case, $L_1 < L_2 < \cdots < L_K$ is assumed, and suppose that $\gamma_1 > \gamma_2 > \cdots > \gamma_K$ for both cases of statistical CSIT and no CSIT.  
All phases of bi-directional cooperative NOMA in the multi-user scenario are explained below:

\subsubsection{Direct Transmission Phase} 
The BS transmits the superpositioned signal to all users.
The received signal at the user $i$ is 
\begin{equation}
r_i = h_i \sum_{k=1}^{K} \sqrt{\gamma_k}s_k + n_i.
\end{equation}

\subsubsection{Channel Information Exchange Phase}
Users exchange their channel information to determine the direction of cooperation.
As we will see in the $j$-th cooperative phase, the cooperation signal is transmitted at the user $i_0$, whose channel gain is the strongest among users $j,\cdots,K$. 
Therefore, users should know the order of channel gains of all users and this phase requires $K$ time slots.
For each slot, a user delivers its channel information to the others via short-range communications.
Note that in the proposed scheme in $K$-user scenario, each time slot is actually reduced for exchange of channel information when the users are crowded.

\textit{Remark}:
As mentioned earlier in Section \ref{subsec:channel_exchange_phase}, the CSI exchange step causes additional delay and overhead, and those penalties grow as $K$ increases. 
However, even without the CSI exchange step, this problem also arises in the existing systems where the BS receives CSI feedbacks.
Large $K$ also causes a huge computational burden for SIC and requires large power budget to enable multiple steps of SIC. 
Therefore, only two or four-user NOMA signaling has been considered in practical system, and the proposed scheme can be applied well with appropriate $K$ for practical scenarios.

\subsubsection{The $j$-th Cooperative Phase}
The cooperation phase consists of $K$ phases. 
Cooperation is performed on signal-by-signal basis, i.e., decoding of $s_j$ is performed in the $j$-th cooperative phase.
Users $j,\cdots,K$ decode $s_j$, and this phase corresponds to one of the SIC steps, especially for users $j+1,\cdots,K$.
Let $V_{i,j}$ be the SINR for decoding $s_j$ at user $i$, where $i\geq j$, and
\begin{equation}
V_{i,j} = \frac{|h_i|^2 \gamma_j}{|h_i|^2\sum_{k=j+1}^{K}\gamma_k + \sigma_n^2}.
\end{equation}
Therefore, cooperation among users $j,\cdots,K$ for improved decoding of $s_j$ is possible in the $j$-th cooperative phase.
Among users $j,\cdots,K$, the system allows one whose channel condition is the best to transmit the cooperation signal. 
Let user $i_0$ be the strongest one, i.e., $i_0 = \underset{j \leq i \leq K}{\arg \max} |h_{i}|^2 = \underset{j \leq i \leq K}{\arg \max} |V_{i,j}|^2$.
Then, user $i$, where $j \leq i \leq K$ and $i \neq i_0$, receives the cooperation signals $c^j_{i}$ to help decoding of $s_j$ from user $i_0$, and user $i_0$ does not receive any cooperation signal.
\begin{equation}
c^j_{i} = g_{i,i_0} s_{j} + n_i,~~\forall i\in \{j, \cdots, K\},~i \neq i_0
\end{equation}
where $g_{i,i_0}$ is channel fading from user $i_0$ to user $i$.
Let $W_{i,i_0}$ be the SINR of the cooperation signal from user $i_0$ to user $i$, as written by
\begin{equation}
W_{i,i_0} = \frac{|g_{i,i_0}|^2}{\sigma_n^2}.
\end{equation}

The $j$-th cooperation step can increase the data rate of user $j$, and/or help other users $j+1,\cdots,K$ to perform SIC better. 
The SINR for decoding $s_j$ at user $i$ is denoted by $Z_{\text{cN},i,j}^{\text{bi}}$ and obtained by
\begin{equation}
Z_{\text{cN},i,j}^{\text{bi}} = 
\begin{cases}
V_{i,j} & i=i_0 \\
\max \{ V_{i,j}, W_{i,i_0} \} & i \neq i_0
\end{cases},~~\forall i \in \{j, \cdots, K\}.
\end{equation}
Therefore, the total SINR for decoding $s_j$ becomes as follows:
\begin{equation}
Z_{\text{cN},j}^{\text{bi}} = \min \{ Z_{\text{cN},j,j}^{\text{bi}}, \cdots, Z_{\text{cN},K,j}^{\text{bi}} \},
\end{equation}
With the assumption that all users are located nearby, the cooperation signal is much stronger than the signal from the BS, i.e., $W_{i,i_0} \gg V_{i,j}$, $Z_{\text{cN},j}^{\text{bi}}$ can be approximated by
\begin{equation}
Z_{\text{cN},j}^{\text{bi}} \approx \tilde{Z}_{\text{cN},j}^{\text{bi}} = \max \{ V_{j,j}, \cdots, V_{K,j} \}.
\end{equation}

On the other hand, in uni-directional cooperative NOMA, the direction of cooperation is already determined, so all users except for user $K$ receive the cooperation signals from user $K$, whose channel variance is the largest in the statistical-CSIT case, or who are arbitrarily determined in the no-CSIT case.
Therefore, the SINRs are given by
\begin{equation}
{Z}_{\text{cN},i,j}^{\text{uni}} = 
\begin{cases}
V_{i,j} & i=K \\
\max \{ V_{i,j}, W_{j,K} \} & i \neq i_0
\end{cases},~~\text{for}~ j\leq i\leq K,
\end{equation}
\begin{align}
{Z}_{\text{cN},j}^{\text{uni}} &= \min \{ {Z}_{\text{cN},j,j}^{\text{uni}}, \cdots, {Z}_{\text{cN},K,j}^{\text{uni}} \} \\
&\approx \tilde{Z}_{\text{cN},j}^{\text{uni}} = V_{K,j}.
\end{align}
Conventional NOMA does not allow any cooperation, so its SINR for decoding $s_j$ is given by
\begin{equation}
Z_{\text{N},j} = \min \{ V_{j,j}, V_{j+1,i}, \cdots, V_{K,i} \}.
\end{equation}

Mathematical analysis of ergodic capacity and outage probability for the multi-user model is omitted here, but we verify the advantages of bi-directional cooperative NOMA in the multi-user model by simulation. 
Section \ref{subsec:capacity_randomlyusers} shows that bi-directional cooperative NOMA gives still better data rates than uni-directional cooperative NOMA, conventional NOMA and OMA for randomly generated multiple users.

\textit{Remark:} Complexity for performing SIC is an important issue in the NOMA system. 
For conventional NOMA, $k-1$ times of SIC processes are required for the $k$-th strongest user and SIC processes of all users are performed independently.
The number of required SIC processes for bi-directional cooperative NOMA is the same as conventional one.
Also, the cooperation phases are performed on signal-by-signal basis, so the SIC step for decoding $s_j$ is performed at every user $i\in\{j,\cdots,K\}$ in parallel.
Therefore, bi-directional cooperative NOMA does not require the additional time slots for SIC processes of different users, as long as the $j$-th cooperation phase is successfully completed after SIC of $s_j$.
On the other hand, the $k$-th strongest user of cooperative NOMA in \cite{NOMA:uni-cN:Ding} should wait for the others with the better channel conditions to finish the SIC processes and to transmit the cooperation signals.

\section{Simulation Results}
\label{sec:simulation}

The simulation model is based on Fig. \ref{fig:bi-cN_model}, and we assume $R=50$.
Without loss of generality, users 1 and 2 are assume to be the non-SIC user and the SIC user, respectively. 
Short-hand notations `Bi-cN', `Uni-cN', and `NOMA' denote bi-directional cooperative NOMA, uni-directional cooperative NOMA, and conventional NOMA, respectively, in the figures.

\subsection{Ergodic Capacity under Statistical-CSIT only}
\label{subsec:ergodicCap_statistical}

First, consider the statistical-CSIT case. 
For bi-directional cooperative NOMA, Algorithms \ref{algo:user_fairness} and \ref{algo:max_sum_rate} are used to find the optimal power allocations for user fairness and sum-rate maximization, respectively. 
For other schemes, uni-directional cooperative NOMA, conventional NOMA and OMA, the optimal power ratios are numerically found.
The ergodic capacity plots versus $d_2$ are obtained with the fixed position of user 1, $d_1=40$.
Since user 2 is the SIC user, $L_1 \leq L_2$, i.e., $d_2 \leq d_1$, in the statistical-CSIT case.

\begin{figure}[t]
	\minipage{0.4\textwidth}
	\includegraphics[width=\linewidth]{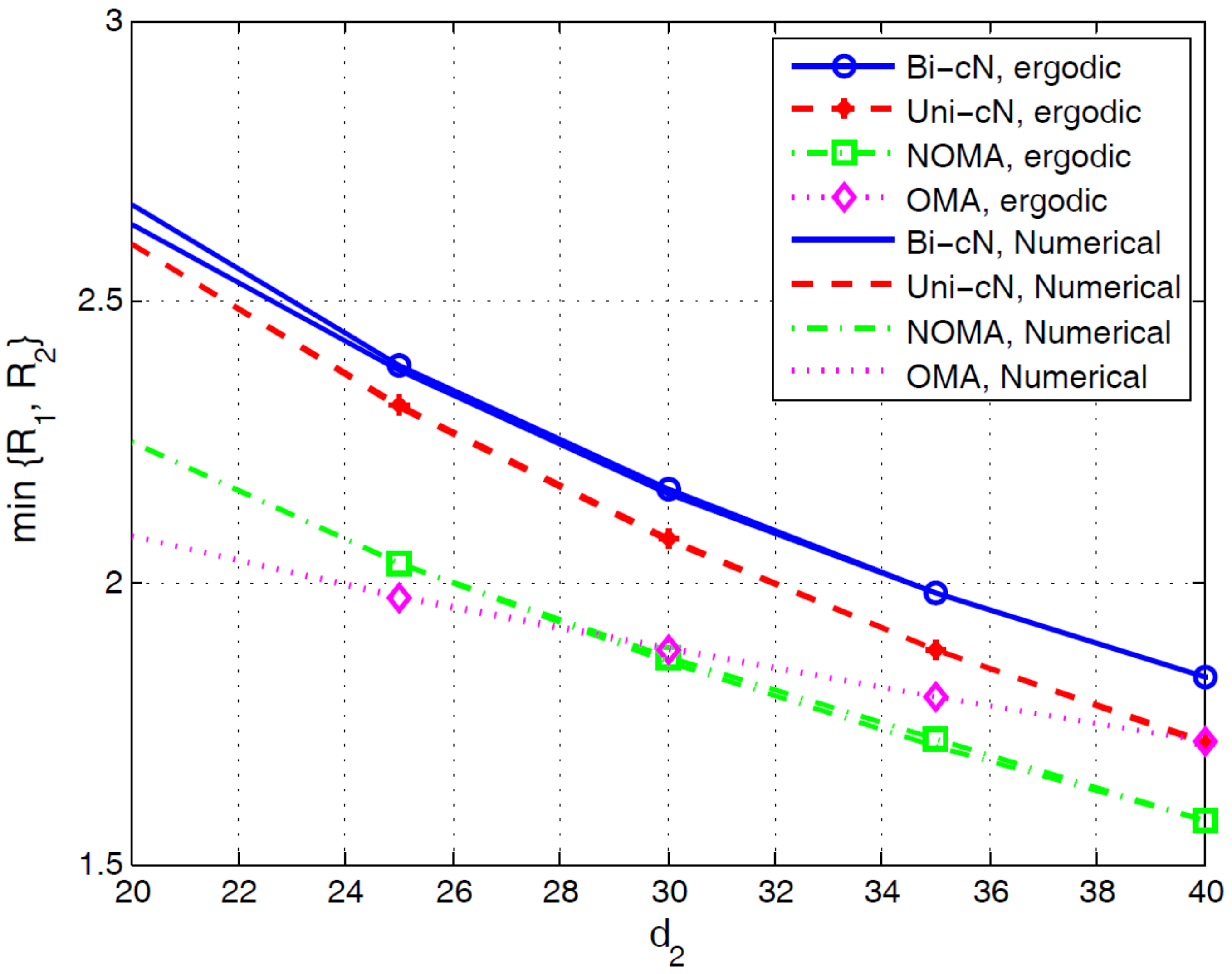}
	\caption{User fairness in statistical-CSIT case, SNR=10dB} \label{fig:Fairness_statisticalCSI_10dB_d1=40}
	\endminipage\hfill
	\minipage{0.4\textwidth}
	\includegraphics[width=\linewidth]{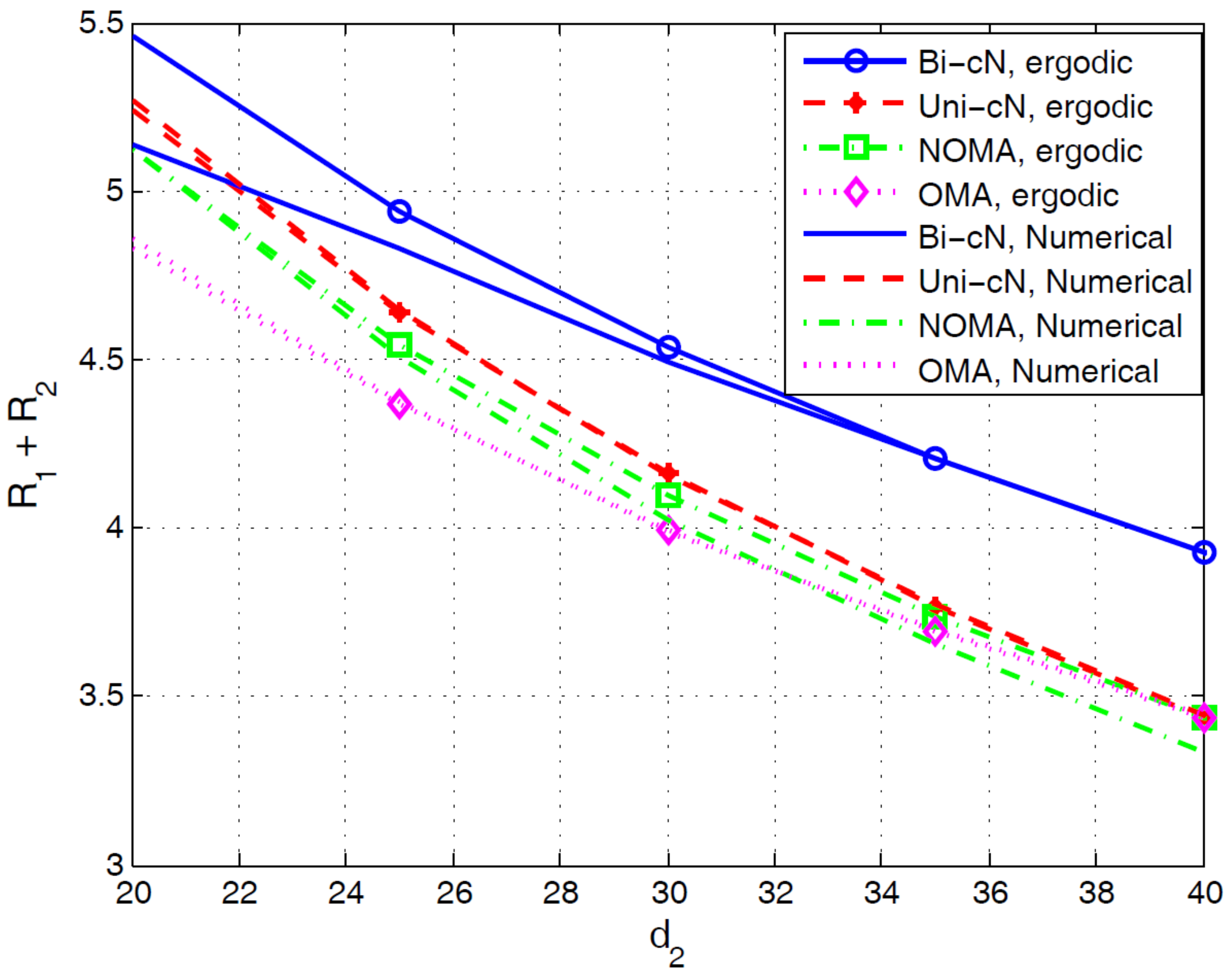}
	\caption{Sum-rate in statistical CSIT-case, SNR=10dB} \label{fig:SumRate_statisticalCSI_10dB_d1=40}
	\endminipage
\end{figure}

Figs. \ref{fig:Fairness_statisticalCSI_10dB_d1=40} and \ref{fig:SumRate_statisticalCSI_10dB_d1=40} show user-fairness and sum-rate performances in the two-user scenario with the transmit SNR of 10 dB, respectively. 
$R_{t,1} = R_{t,2} = 0.8$ is assumed for sum-rate results.
Each figure includes both ergodic capacity and numerically obtained data rates, and we can easily see that both are almost the same as $d_2$ approaches to $d_1=40$. 
This indicates the assumption that the users are located nearby is reliable, when $d_1 \approx d_2$.
When $d_2$ is small, the ergodic capacity and numerically obtained sum-rate of bi-directional cooperative NOMA are somewhat different, because Algorithm \ref{algo:max_sum_rate} is suboptimal, as mentioned earlier.
In Figs. \ref{fig:Fairness_statisticalCSI_10dB_d1=40} and \ref{fig:SumRate_statisticalCSI_10dB_d1=40}, bi-directional cooperative NOMA gives better performances of both user fairness and sum-rate than other schemes.
As $d_2$ approaches to $d_1$, the channel gain difference between the users decreases, so the capacity gain of NOMA compared to OMA also decreases.
Therefore, uni-directional cooperative NOMA has the same performances as OMA and conventional NOMA becomes even worse than OMA, when $d_1 \approx d_2$.  
Whereas, bi-directional cooperative NOMA is still better than OMA, so it can be said that bi-directional cooperative NOMA is useful even when the channel gain difference of users is not large.

\subsection{Ergodic Capacity under No-CSIT}
\label{subsec:ergodicCap_noCSIT}

In this subsection, we consider the situation where the BS does not know users' CSI at all. 
Again, the two-user scenario is considered with $d_1=40$.
Since no CSI is available at the BS, the BS arbitrarily determines users 1 and 2 as the non-SIC user and the SIC user, respectively, so $d_2>d_1$ is possible.
It is impossible to find the optimal power allocation, so the fixed power allocation is used.
$\gamma_1 = 0.8$ is assumed for NOMA schemes, and $\gamma_1 = 0.5$ and $\alpha_1 = 0.5$ are used for OMA, because fair power allocation is preferable for OMA without any CSI. 
Also, only sum-rate performances are investigated in the no-CSIT case, because the trends of user-fairness performances depend largely on the power allocation ratios.

\begin{figure} [h!]
	\centering
	\includegraphics[width=0.42\textwidth]{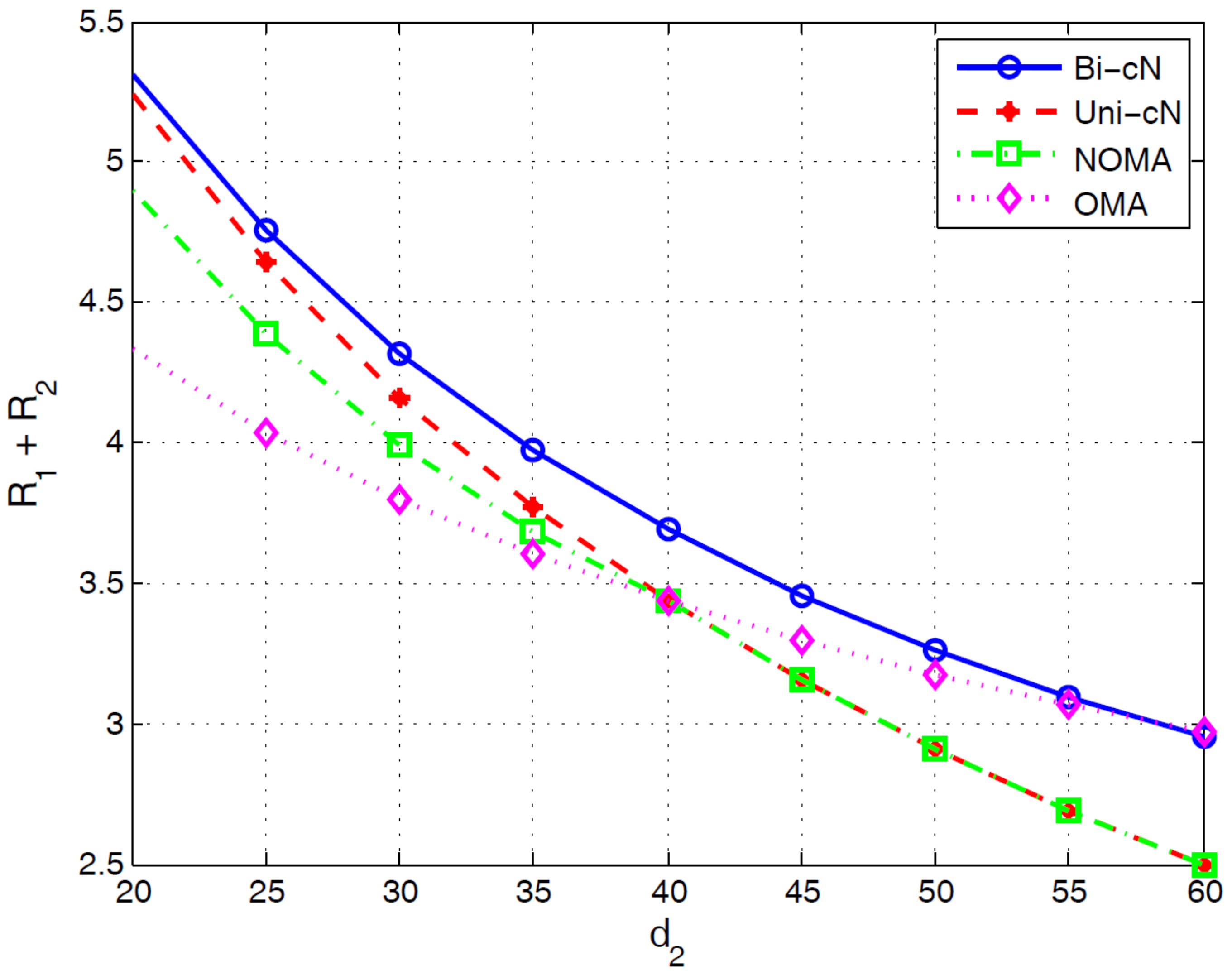}
	\caption{Sum-rate in no-CSIT case, SNR=10dB} \label{fig:SumRate_noCSI_10dB_d1=40}
\end{figure}

Fig. \ref{fig:SumRate_noCSI_10dB_d1=40} shows the sum-rate graphs with no CSIT and a transmit SNR of 10 dB.
In Fig. \ref{fig:SumRate_noCSI_10dB_d1=40}, it is easily noticeable that bi-directional cooperative NOMA gives better capacity than other NOMA schemes and OMA.
When $d_2<d_1$, capacity gains of bi-directional cooperative NOMA over other NOMA schemes are not large, but its gain over OMA is large due to SIC.
When $d_2 \geq d_1$, since the channel gain of the SIC user usually becomes smaller than that of the non-SIC user, the advantage of NOMA and SIC vanishes, so uni-directional cooperative NOMA and conventional NOMA become worse than OMA. 
However, bi-directional cooperative NOMA still has a capacity gain compared to OMA even when $d_2 \geq d_1$. 
These results are consistent with Theorems \ref{thm:ergodic_comparison_NOMA} and \ref{theorem:ergodic_inequality_oma}. 
Especially when $d_1=d_2=40$, bi-directional cooperative NOMA shows a capacity increase of 15\% compared to OMA.
As $d_2$ increases much, bi-directional cooperative NOMA would give a smaller sum-rate than OMA, but its gap is relatively small, compared to the region of $d_2 \leq d_1$.

\subsection{Outage Probability}

\begin{figure}[t]
	\minipage{0.4\textwidth}
	\includegraphics[width=\linewidth]{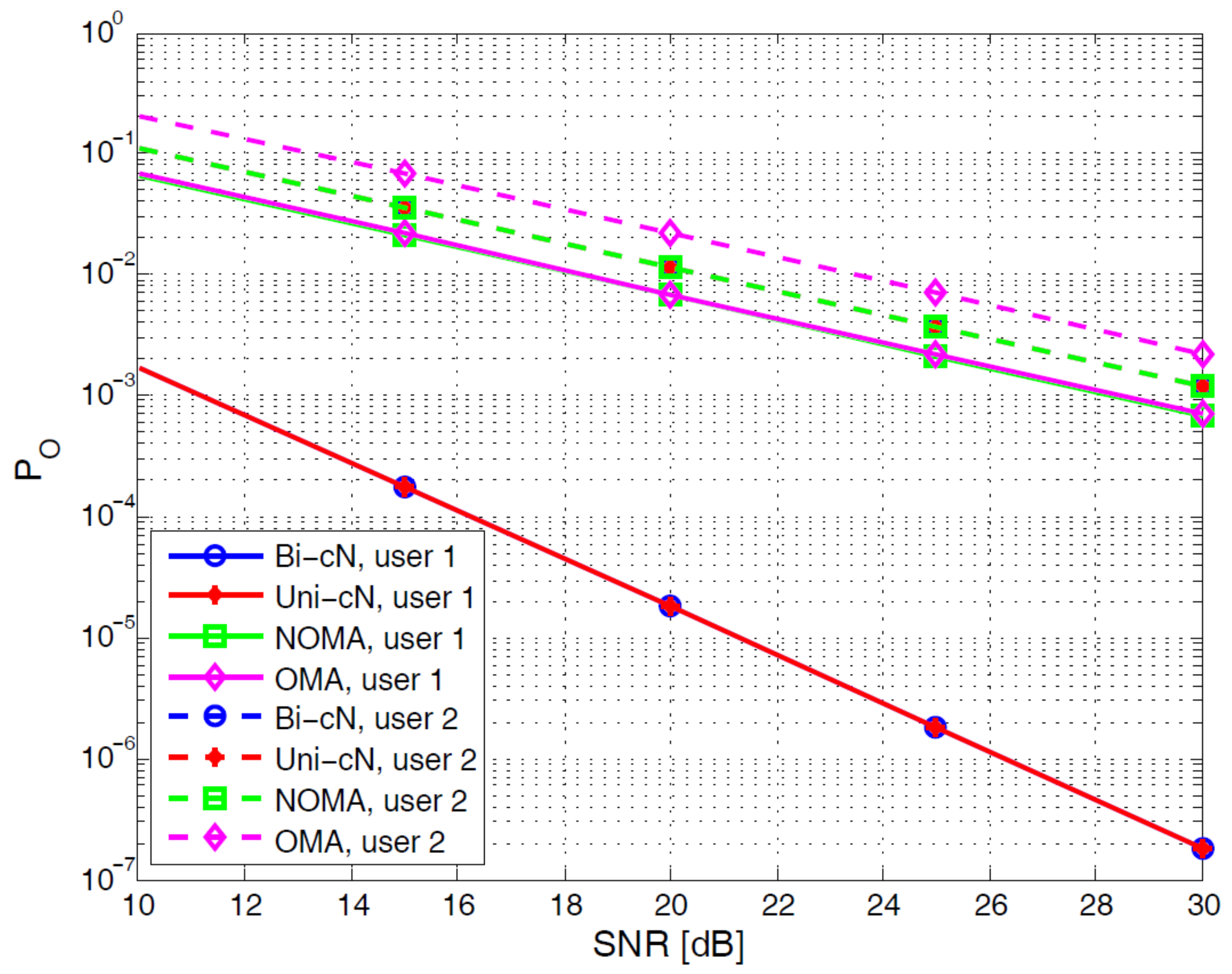}
	\caption{Outage probabilities when $d_1=40, d_2=20, R_{t1}=0.7, R_{t2}=1.5$}
	\label{fig:Outage_d1=40_d2=20_Rt1=07_Rt2=15_p1=075}
	\endminipage\hfill
	\minipage{0.4\textwidth}
	\includegraphics[width=\linewidth]{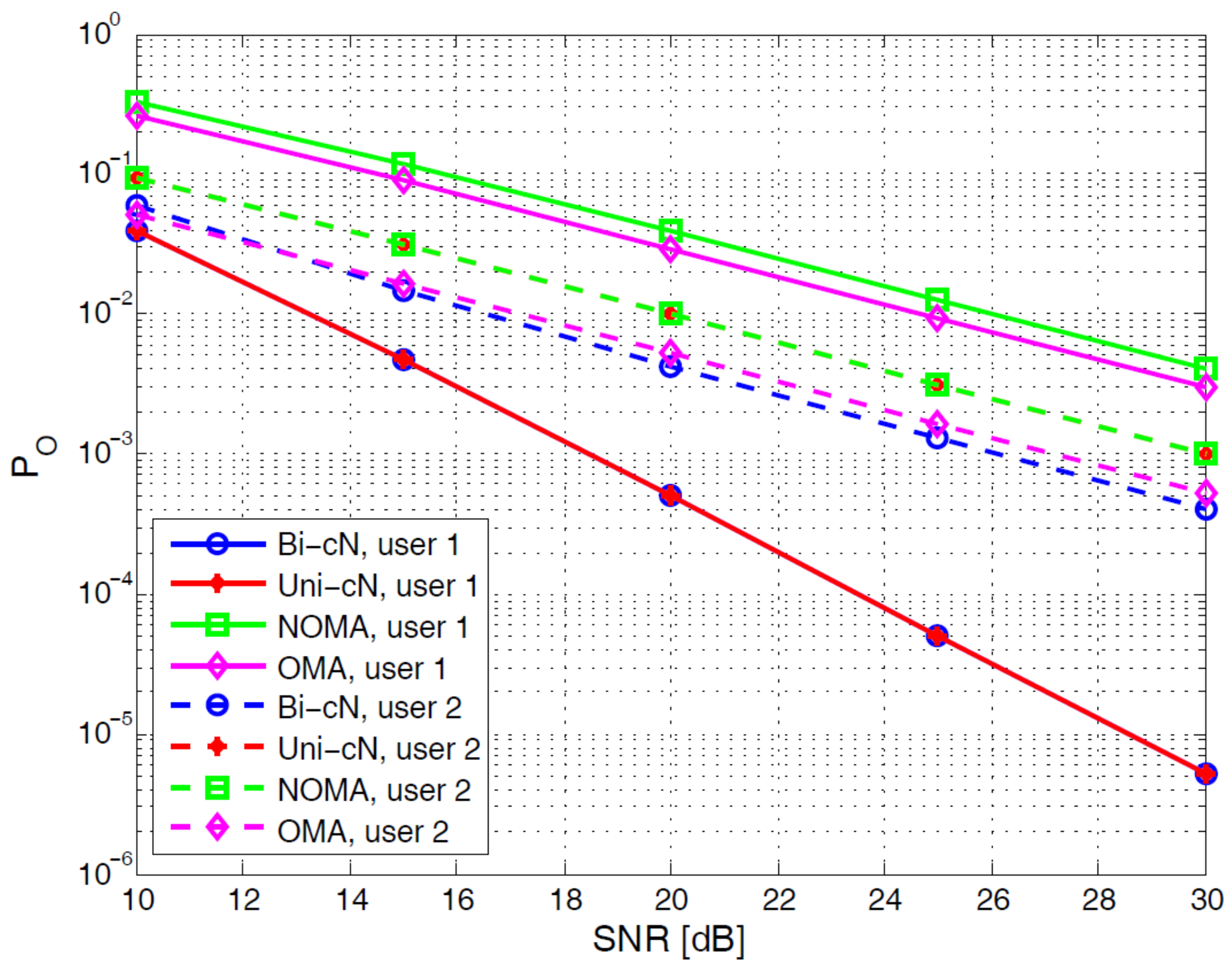}
	\caption{Outage probabilities when $d_1=40, d_2=20, R_{t1}=1.5, R_{t2}=0.7$}
	\label{fig:Outage_d1=40_d2=20_Rt1=15_Rt2=07_p1=075}
	\endminipage
\end{figure}

\begin{figure}[t]
	\minipage{0.4\textwidth}
	\includegraphics[width=\linewidth]{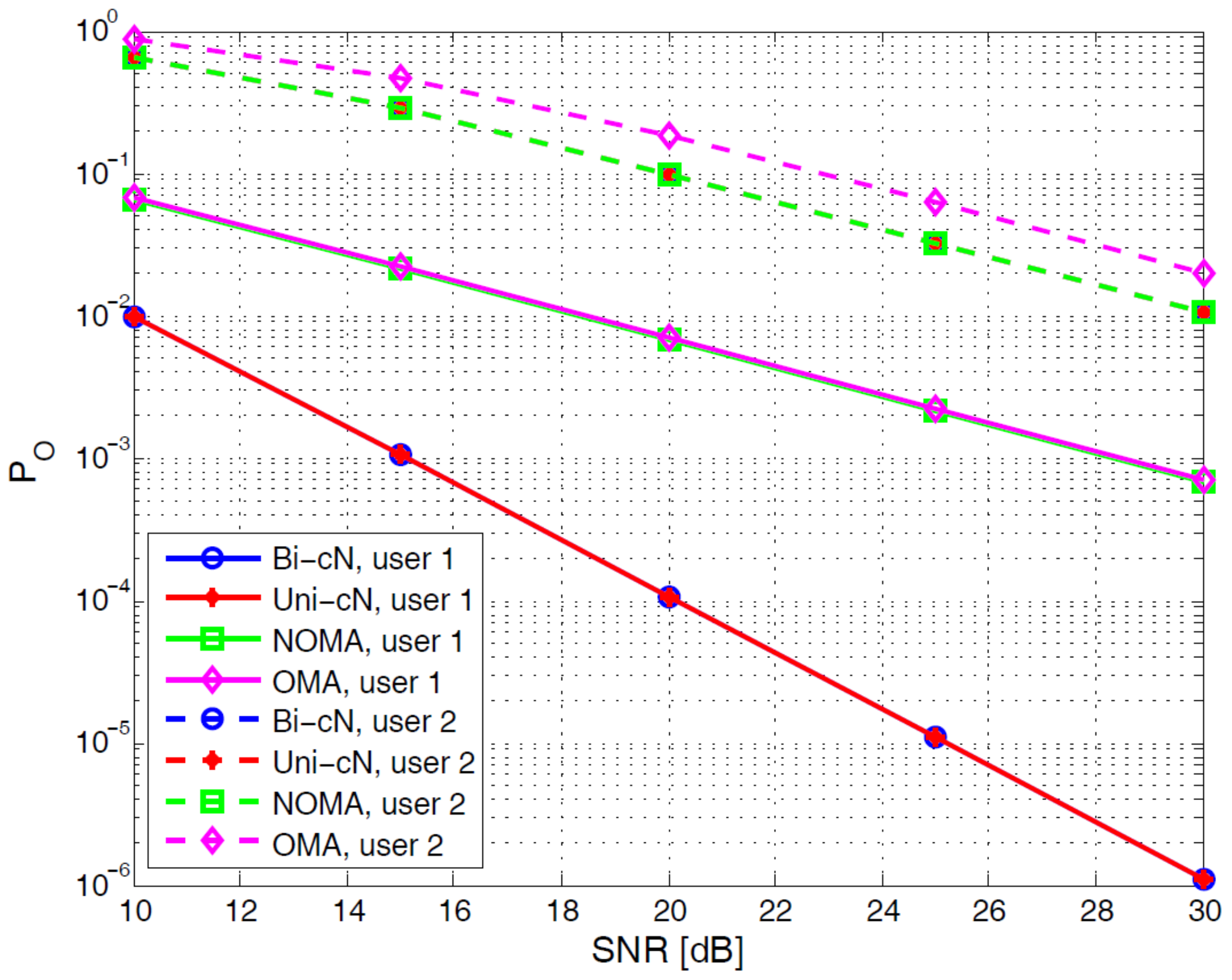}
	\caption{Outage probabilities when $d_1=40, d_2=60, R_{t1}=0.7, R_{t2}=1.5$}
	\label{fig:Outage_d1=40_d2=60_Rt1=07_Rt2=15_p1=075}
	\endminipage\hfill
	\minipage{0.4\textwidth}
	\includegraphics[width=\linewidth]{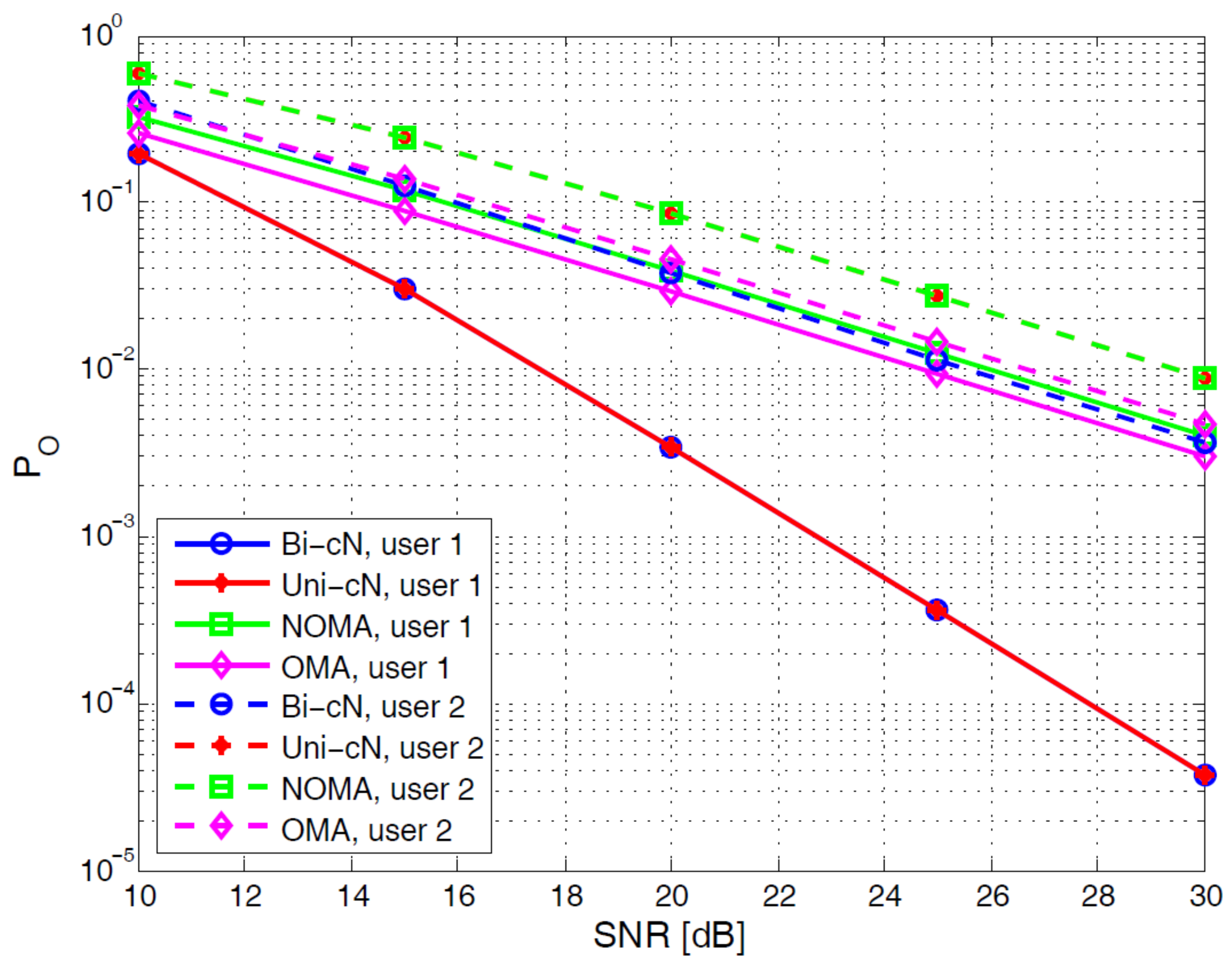}
	\caption{Outage probabilities when $d_1=40, d_2=60, R_{t1}=1.5, R_{t2}=0.7$}
	\label{fig:Outage_d1=40_d2=60_Rt1=15_Rt2=07_p1=075}
	\endminipage
\end{figure}

As shown in Section \ref{sec:outage_prob}, the outage probability is computed by channel variances and power allocation ratios.
Power allocation and $R_{t,1}$ are appropriately chosen to satisfy $\frac{\gamma_1}{\gamma_2} > \epsilon_1$, and $\gamma_1=0.75$ is assumed here.
Figs. \ref{fig:Outage_d1=40_d2=20_Rt1=07_Rt2=15_p1=075} and \ref{fig:Outage_d1=40_d2=20_Rt1=15_Rt2=07_p1=075} give outage probability performances when $d_1=40$ and $d_2=20$.
Fig. \ref{fig:Outage_d1=40_d2=20_Rt1=07_Rt2=15_p1=075} assumes $R_{t,1}=0.7$ and $R_{t,2}=1.5$ and Fig. \ref{fig:Outage_d1=40_d2=20_Rt1=15_Rt2=07_p1=075} is obtained with $R_{t,1}=1.5$ and $R_{t,2}=0.7$.
In other words, Fig. \ref{fig:Outage_d1=40_d2=20_Rt1=07_Rt2=15_p1=075} satisfies the condition of $\frac{\epsilon_2}{\gamma_2} > \frac{\epsilon_1}{\gamma_1 - \epsilon_1\gamma_2}$, so bi- and uni-directional cooperative NOMA schemes show exactly the same outage probabilities for both users. 
Also, it is easily noted by the slopes of graphs that bi- and uni-directional cooperative NOMA schemes provide a better diversity order than conventional NOMA and OMA for user 1, but not for user 2.
Difference in diversity order is also observed in Fig. \ref{fig:Outage_d1=40_d2=20_Rt1=15_Rt2=07_p1=075}, and two cooperative NOMA schemes still have the same outage probability of user 1.
However, since $\frac{\epsilon_2}{\gamma_2} < \frac{\epsilon_1}{\gamma_1 - \epsilon_1\gamma_2}$ is satisfied in Fig. \ref{fig:Outage_d1=40_d2=20_Rt1=15_Rt2=07_p1=075}, the outage probability of user 2 of bi-directional cooperative NOMA has a power gain compared to uni-directional cooperative NOMA.

Figs. \ref{fig:Outage_d1=40_d2=60_Rt1=07_Rt2=15_p1=075} and \ref{fig:Outage_d1=40_d2=60_Rt1=15_Rt2=07_p1=075} give outage probability performances obtained with $d_1=40$ and $d_2=60$.
Since $L_1>L_2$ in Figs. \ref{fig:Outage_d1=40_d2=60_Rt1=07_Rt2=15_p1=075} and \ref{fig:Outage_d1=40_d2=60_Rt1=15_Rt2=07_p1=075}, those results only correspond to the no-CSIT case.
The plots in Figs. \ref{fig:Outage_d1=40_d2=60_Rt1=07_Rt2=15_p1=075} and \ref{fig:Outage_d1=40_d2=60_Rt1=15_Rt2=07_p1=075} show almost same trends with those in Figs. \ref{fig:Outage_d1=40_d2=20_Rt1=07_Rt2=15_p1=075} and \ref{fig:Outage_d1=40_d2=20_Rt1=15_Rt2=07_p1=075}, except for a little bit of power gain differences.
The power gain of bi-directional cooperative NOMA over uni-directional cooperative NOMA and the diversity order gains over conventional NOMA and OMA are still guaranteed even when $d_1=40$ and $d_2=60$.

\begin{figure} [h!]
	\centering
	\includegraphics[width=0.48\textwidth]{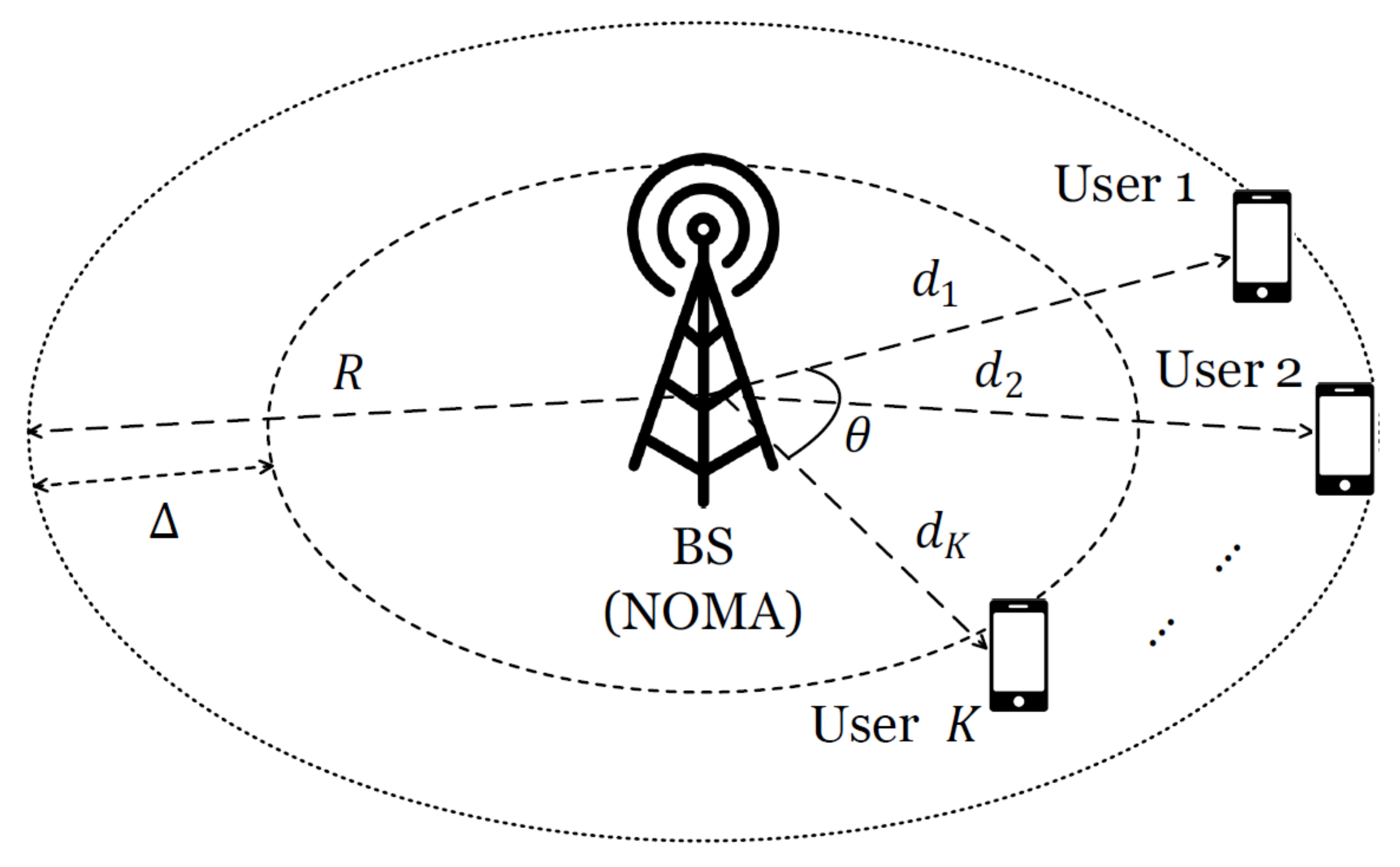}
	\caption{Cellular model of randomly generated multiple users}
	\label{fig:RandomUserModel}
\end{figure}

\subsection{Capacity of Randomly Generated Users}
\label{subsec:capacity_randomlyusers}

This section considers the cellular model of randomly positioned multiple users, as shown in Fig. \ref{fig:RandomUserModel}. 
In Fig. \ref{fig:RandomUserModel}, $K$ users are uniformly placed in the outer ring of the cell of radius $R=50$, i.e., the region between inner and outer circles whose radii are $R-\Delta$ and $R$, respectively, so $d_k \in [R-\Delta,R],~\forall k=\{1,\cdots,K\}$.
When $\Delta$ is very small, the only cell-edge users are chosen, and as $\Delta$ increases, almost all of the cell region is covered.
Also, $K$ users are separated by an angle smaller than $\theta$ to guarantee that all users are close to one another. 
It makes the exchange of CSI among users easier and the cooperation more helpful.
We only presents an optimal power allocation rule for the two-user model, so optimal power allocations for $K$ users are numerically obtained in the statistical-CSIT case.
On the other hand, in the case of no CSIT, fixed power allocation is applied.
Assume that $\gamma_{k} = 2\gamma_{k+1}$ for NOMA schemes, and $\gamma_1 = \cdots = \gamma_K = 0.5$ and $\alpha_1 = \cdots = \alpha_K = 0.5$ are used for OMA.
Similar to Section \ref{subsec:ergodicCap_noCSIT}, the sum-rate performance is only considered in the no-CSIT case.

\begin{figure}[t]
	\minipage{0.4\textwidth}
	\includegraphics[width=\linewidth]{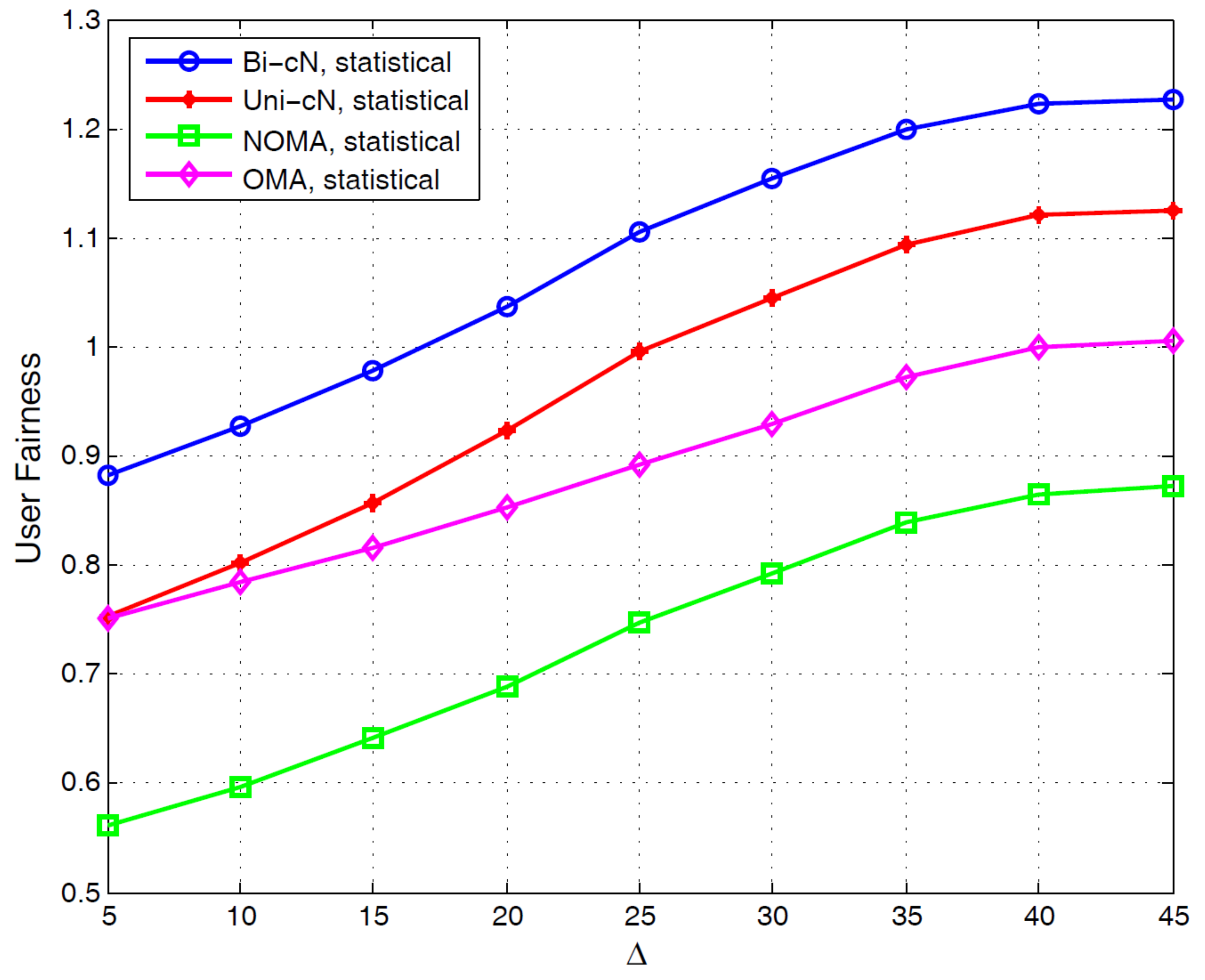}
	\caption{Capacity for user fairness in cellular model of randomly generated users}
	\label{fig:CapacityRandomEdgeUser_UserFairness}
	\endminipage\hfill
	\minipage{0.4\textwidth}
	\includegraphics[width=\linewidth]{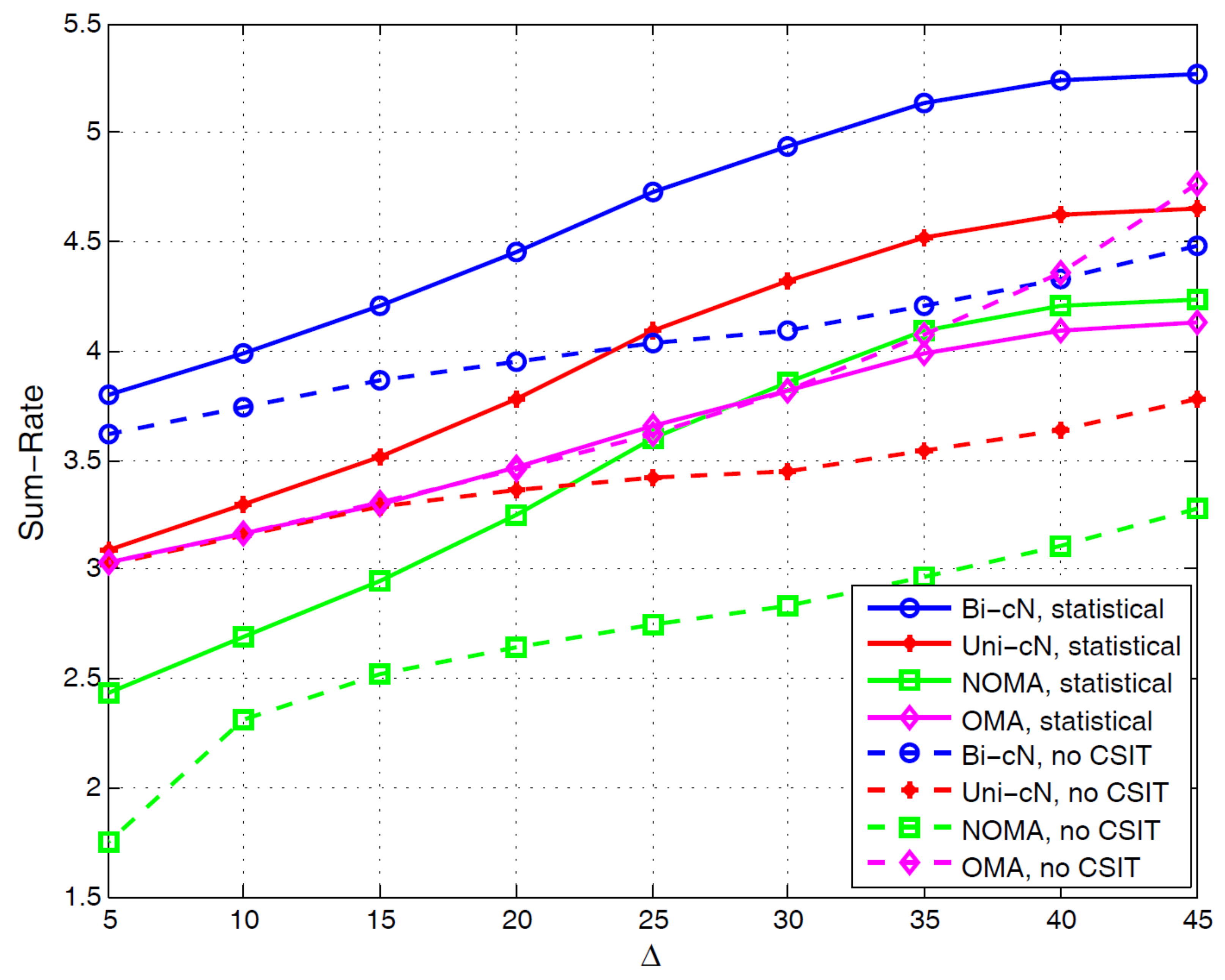}
	\caption{Capacity for sum-rate in cellular model of randomly generated users}
	\label{fig:CapacityRandomEdgeUser_SumRate}
	\endminipage
\end{figure}

Figs. \ref{fig:CapacityRandomEdgeUser_UserFairness} and \ref{fig:CapacityRandomEdgeUser_SumRate} show plots of data rates versus $\Delta$ obtained from user fairness and sum-rate problems, respectively, with $K=4$ randomly located users. 
Solid and dashed graphs correspond to the cases of statistical CSIT and no CSIT, respectively.
All graphs show increasing trends over $\Delta$, because users are likely to have stronger channel gains with a larger $\Delta$.
We can easily see that bi-directional cooperative NOMA gives the best data rates among comparison schemes both in the statistical-CSIT and the no-CSIT cases.
Especially in the statistical-CSIT case, when $\Delta$ is small, channel gain differences among users are not large enough to show the advantage of NOMA schemes compared to OMA, so the performance of uni-directional cooperative NOMA is similar to that of OMA, and conventional NOMA is worse than OMA.
On the other hand, bi-directional cooperative NOMA is still better than OMA even with small $\Delta$, in terms of both user fairness and sum-rate. 

The sum-rate performances in the no-CSIT case are much worse than those in the statistical-CSIT case as $\Delta$ increases. 
The reason is that when $\Delta$ is large, the situation in which users with weaker channel conditions perform SIC for decoding the signal of users with stronger channel gains so the advantage of NOMA vanishes happens frequently.
In this situation, smaller power is allocated to weaker user than stronger one, so uni-directional cooperative NOMA and conventional NOMA give worse sum-rates than OMA in the no-CSIT case.
However, the sum-rates of bi-directional cooperative NOMA are still better than other NOMA schemes as well as OMA in the most values of $\Delta$, even when the BS does not have accurate knowledge of users' CSI and thus arbitrarily schedules the users for signal transmission.
This means that bi-directional cooperative NOMA is useful when there is little need for channel tracking or elaborate user scheduling.

\section{Conclusion}
\label{sec:conclusion}

This paper proposes bi-directional cooperative NOMA, in which NOMA users cooperate with each other by channel information exchange, with statistical CSIT only or no CSIT.
Performance analysis has been conducted in terms of ergodic capacity and outage probability. 
The closed-form ergodic capacity of bi-directional cooperative NOMA in the two-user model is derived, and it is shown to be better than those of uni-directional cooperative NOMA, conventional NOMA, and OMA even when the scheduled users have similar channel gains, under only statistical CSI or no CSI at the BS.
Based on the ergodic capacity, algorithms to find the optimal power allocations are presented for user fairness and max-sum-rate problems.
In addition, we show that the outage probability of the SIC user of bi-directional cooperative NOMA has a power gain over that of uni-directional cooperative NOMA, and bi- and uni-directional cooperative NOMA schemes have a diversity gain over conventional NOMA and OMA.
Also, the multi-user model of bi-directional cooperative NOMA is presented by using the cooperations among users on signal-by-signal basis.
Simulation results verify the above performance analyses, and that bi-directional cooperative NOMA works well with multiple users in statistical- and no-CSIT cases.
The proposed system is beneficial in some important practical scenarios: a highly crowded stadium in which many users experience similar channel gains, and an IoT environment in which an inexpensive transmitter should serve a massive number of machine-type devices but cannot handle substantial processing tasks, so not enough CSI is available at the BS.




\appendices
\section{Proof of Theorem \ref{thm:closed_ergodic_bi-cN}}
\label{appendix:thm1}

$|h_1|^2$ and $|h_2|^2$ are chi-square distributed with variances of $L_1/2$ and $L_2/2$, respectively, so $|h_i|^2 = L_i X_i / 2$, for $i=\{1,2\}$, where $X_i$ is a chi-square random variable of unit variance.
According to (\ref{eq:bi-cN:R1-near-user}), 
\begin{eqnarray}
\mathbb{E}[\tilde{R}_{cN,1}^{bi}] &=& \iint p(x_1,x_2) \tilde{R}_{cN,1}^{bi} \mathrm{d}x_1 \mathrm{d}x_2 \\
&=& \iint_{|h_1|^2 > |h_2|^2} p(x_1,x_2) Z_1 \mathrm{d}x_1 \mathrm{d}x_2 + \iint_{|h_1|^2 < |h_2|^2 } p(x_1,x_2) Z_{2,SIC} \mathrm{d}x_1 \mathrm{d}x_2.
\label{eq:exp_cN1}
\end{eqnarray}

Since $X_1$ and $X_2$ are independent, $p(x_1, x_2) = p(x_1)p(x_2)$, the first term of (\ref{eq:exp_cN1}) becomes
\begin{eqnarray}
&&\int_{0}^{\infty} \int_{0}^{\frac{L_1}{L_2}x_1} \frac{1}{4} e^{-\frac{x_1}{2}} e^{-\frac{x_2}{2}} \bigg\{ \log_2\Big(1+ \frac{\gamma_1 L_1 x_1}{\gamma_2 L_1 x_1 + 2\sigma_n^2}\Big) \bigg\} \mathrm{d}x_1 \mathrm{d}x_2 \\
&=& \int_{0}^{\infty} \frac{1}{2} (e^{-\frac{x_1}{2}} - e^{-\frac{L_1+L_2}{2L_2}x_1}) \bigg\{ \log_2\Big(1+ \frac{L_1 x_1}{2\sigma_n^2}\Big) - \log_2\Big(1+ \frac{\gamma_2 L_1 x_1}{2\sigma_n^2}\Big) \bigg\} \mathrm{d}x_1 \\
&=& C_1\Big(\frac{L_1}{\sigma_n^2}\Big) - C_1\Big(\frac{\gamma_2 L_1}{\sigma_n^2}\Big) - \frac{L_2}{L_1 + L_2} \bigg\{ C_1\Big(\frac{L_1 L_2}{(L_1+L_2)\sigma_n^2}\Big) - C_1\Big(\frac{\gamma_2 L_1 L_2}{(L_1+L_2)\sigma_n^2}\Big) \bigg\}.
\label{eq:temp1}
\end{eqnarray}
The last equation (\ref{eq:temp1}) holds by Lemma \ref{lemma:integral_chi}.

Likewise, the second term of (\ref{eq:exp_cN1}) is
\begin{eqnarray}
&&\int_{0}^{\infty} \int_{0}^{\frac{L_2}{L_1}x_2} \frac{1}{4} e^{-\frac{x_1}{2}} e^{-\frac{x_2}{2}} \log_2 \Big( 1 + \frac{\gamma_1 L_2 x_2}{\gamma_2 L_2 x_2 + 2\sigma_n^2} \Big) \mathrm{d}x_1 \mathrm{d}x_2 \\
&=& C_1 \Big(\frac{L_2}{\sigma_n^2}\Big) - C_1\Big(\frac{\gamma_2 L_2 }{\sigma_n^2}\Big) - \frac{L_1}{L_1 + L_2} \bigg\{ C_1 \Big(\frac{L_1 L_2}{(L_1+L_2)\sigma_n^2}\Big) - C_1\Big(\frac{\gamma_2 L_1 L_2}{(L_1+L_2)\sigma_n^2}\Big) \bigg\}
\end{eqnarray}

Therefore, 
\begin{equation}
\mathbb{E}[\tilde{R}_{cN,1}^{bi}] = C_1\Big(\frac{L_1}{\sigma_n^2}\Big) - C_1\Big(\frac{\gamma_2 L_1}{\sigma_n^2}\Big) + C_1\Big(\frac{L_2}{\sigma_n^2}\Big) - C_1\Big(\frac{\gamma_2 L_2}{\sigma_n^2}\Big) - C_1\Big(\frac{L_1 L_2}{(L_1+L_2)\sigma_n^2}\Big) + C_1\Big(\frac{\gamma_2 L_1 L_2}{(L_1+L_2)\sigma_n^2}\Big)
\end{equation}

Next, 
\begin{eqnarray}
\mathbb{E}[R_{cN,2}^{bi}] &=& \int_{0}^{\infty} \int_{0}^{\infty} \frac{1}{4} e^{-\frac{x_1}{2}} e^{-\frac{x_2}{2}} \log_2 \Big(1+\frac{\gamma_2 L_2 x_2}{2\sigma_n^2}\Big) \mathrm{d}x_1 \mathrm{d}x_2 \\
&=& \int_{0}^{\infty} \frac{1}{2} e^{-\frac{x_2}{2}} \log_2 \Big(1+\frac{\gamma_2 L_2 x_2}{2\sigma_n^2}\Big) \mathrm{d}x_2 = C_1 \Big(\frac{\gamma_2 L_2}{\sigma_n^2}\Big).
\label{eq:temp2}
\end{eqnarray}
Equation (\ref{eq:temp2}) holds by Lemma \ref{lemma:integral_chi}.
Thus, closed-form ergodic capacity (\ref{eq:closed_ergodic_cN_bi}) can be obtained.

\section{Proof of Theorem \ref{thm:ergodic_comparison_NOMA}}
\label{appendix:thm2}

\begin{eqnarray}
\mathbb{E}[\tilde{R}^{bi}_{cN,1}] &=& \mathbb{E}\bigg[\max \Big\{\log_2 \Big(1 + \frac{|h_1|^2 \gamma_1}{|h_1|^2 \gamma_2 +\sigma_n^2}\Big),~\log_2\Big(1 + \frac{|h_2|^2 \gamma_1}{|h_2|^2 \gamma_2 +\sigma_n^2}\Big)\Big\}\bigg]\\
&\geq& \max \bigg\{ \mathbb{E}\Big[\log_2 \Big(1 + \frac{|h_1|^2 \gamma_1}{|h_1|^2 \gamma_2 +\sigma_n^2}\Big)\Big],~\mathbb{E}\Big[\log_2\Big(1 + \frac{|h_2|^2 \gamma_1}{|h_2|^2 \gamma_2 +\sigma_n^2}\Big)\Big] \bigg\} \\
&=& \max \bigg\{ C_1 \Big(\frac{L_1}{\sigma_n^2}\Big) - C_1 \Big(\frac{\gamma_2L_1}{\sigma_n^2}\Big),~C_1 \Big(\frac{L_2}{\sigma_n^2}\Big) - C_1 \Big(\frac{\gamma_2L_2}{\sigma_n^2}\Big) \bigg\}
\label{eq:R1_bi-cN_closed_lowerbound}
\end{eqnarray}
The last equation (\ref{eq:R1_bi-cN_closed_lowerbound}) holds by Lemma \ref{lemma:expected_log_chi}. 

For uni-directional cooperative NOMA system, according to \ref{lemma:expected_log_chi},
\begin{eqnarray}
\mathbb{E}[\tilde{R}_{cN,1}^{uni}] &=& \mathbb{E} \bigg[\log_2\Big(1 + \frac{|h_2|^2 \gamma_1}{|h_2|^2 \gamma_2 +\sigma_n^2}\Big)\bigg]= C_1 \Big(\frac{L_2}{\sigma_n^2}\Big) - C_1 \Big(\frac{\gamma_2L_2}{\sigma_n^2}\Big)
\label{eq:R1_uni-cN_closed}
\end{eqnarray}

Likewise, the upper bound on ergodic capacity of conventional NOMA system becomes
\begin{eqnarray}
\mathbb{E}[R_{N,1}] &=& \mathbb{E}\bigg[\min \Big\{\log_2 \Big(1 + \frac{|h_1|^2 \gamma_1}{|h_1|^2 \gamma_2 +\sigma_n^2}\Big),~\log_2\Big(1 + \frac{|h_2|^2 \gamma_1}{|h_2|^2 \gamma_2 +\sigma_n^2}\Big)\Big\}\bigg]\\
&\leq& \min \bigg\{ \mathbb{E}\Big[\log_2 \Big(1 + \frac{|h_1|^2 \gamma_1}{|h_1|^2 \gamma_2 +\sigma_n^2}\Big)\Big],~\mathbb{E}\Big[\log_2\Big(1 + \frac{|h_2|^2 \gamma_1}{|h_2|^2 \gamma_2 +\sigma_n^2}\Big)\Big] \bigg\} \\
&=& \min \bigg\{ C_1 \Big(\frac{L_1}{\sigma_n^2}\Big) - C_1 \Big(\frac{\gamma_2L_1}{\sigma_n^2}\Big),~C_1 \Big(\frac{L_2}{\sigma_n^2}\Big) - C_1 \Big(\frac{\gamma_2L_2}{\sigma_n^2}\Big) \bigg\}
\label{eq:R1_noma_closed_upperbound}
\end{eqnarray}

Note that the data rates of $s_2$ of three schemes are all the same, 
\begin{equation}
\mathbb{E}[R_{cN,2}^{bi}] = \mathbb{E}[R_{cN,2}^{uni}] = \mathbb{E}[R_{N,2}] = \mathbb{E}\bigg[\log_2\Big(1+\frac{|h_2|^2 \gamma_2 }{\sigma_n^2}\Big)\bigg] = C_1 \Big(\frac{\gamma_2 L_2}{\sigma_n^2}\Big)
\label{eq:R2_closed}
\end{equation}
Since $\max\{x,y\} \geq x \geq \min \{x,y\},~\forall x,y \in \mathbb{R}$, $\mathbb{E}[\tilde{R}^{bi}_{cN}] \geq \mathbb{E}[\tilde{R}^{uni}_{cN}] \geq \mathbb{E}[R_{N}]$ is satisfied according to (\ref{eq:R1_bi-cN_closed_lowerbound}), (\ref{eq:R1_uni-cN_closed}), (\ref{eq:R1_noma_closed_upperbound}), and (\ref{eq:R2_closed}).

\section{Proof of Theorem \ref{theorem:ergodic_inequality_oma}}
\label{appendix:thm3}

First, concavity of $C_1(x)$ for $x>1$ is proved.
Differentiating $C_1(x)$ twice, 
%
\begin{equation}
\frac{\mathrm{d}^2}{\mathrm{d}x^2} C_1(x) = \Big( \frac{2}{x^3}+\frac{1}{x^4} \Big) C_1(x) - \frac{1}{\ln 2} \Big( \frac{1}{x^3} + \frac{1}{x^2} \Big)
\end{equation}

Let $f(x) = x^4 \frac{\mathrm{d}^2}{\mathrm{d}x^2} C_1(x) = (2x+1)C_1(x) - \frac{1}{\ln 2}(x^2+x)$.
Then, 
\begin{eqnarray}
x^2 \frac{\mathrm{d}}{\mathrm{d}x}f(x) &=& (2x^2-x-1)\Big( C_1(x) - \frac{x}{\ln2} \Big) - x C_1(x) \\
&<& \frac{1}{\ln 2} (2x^2-x-1)(\ln(1+x)-x) - xC_1(x)
\label{eq:temp4}
\end{eqnarray}
Since $2x^2-x-1 > 0$ when $x>1$ and $x > \ln (1+x)$, equation (\ref{eq:temp4}) is smaller than 0 when $x>1$.
$f(1)<0$ and $\frac{\mathrm{d}}{\mathrm{d}x}f(x) < 0$, so $f(x)$ is a strictly decreasing function of $x$ for $x>1$.
Therefore, $C_1(x)$ is a strictly concave function of $x$ for $x>1$.

Then,
\begin{eqnarray}
\mathbb{E}[\tilde{R}_{cN}^{bi}] - \mathbb{E}[R_O] &=& C_1 \Big(\frac{L}{\sigma_n^2}\Big) - C_1 \Big(\frac{\gamma_2 L}{\sigma_n^2}\Big) - \bigg\{ C_1 \Big(\frac{L}{2\sigma_n^2}\Big) - C_1 \Big(\frac{\gamma_2 L}{2\sigma_n^2}\Big) \bigg\} \nonumber \\
&&~~+C_1 \Big(\frac{L}{\sigma_n^2}\Big) - \alpha_1 C_1 \Big(\frac{\gamma_1 L}{\alpha_1 \sigma_n^2}\Big) - \alpha_2 C_1 \Big(\frac{\gamma_2 L}{\alpha_2 \sigma_n^2}\Big)
\end{eqnarray}
By Lemma \ref{lemma:Cx-Cax}, 
\begin{equation}
C_1 \Big(\frac{L}{\sigma_n^2}\Big) - C_1 \Big(\frac{\gamma_2 L}{\sigma_n^2}\Big) \geq C_1 \Big(\frac{L}{2\sigma_n^2}\Big) - C_1 \Big(\frac{\gamma_2 L}{2\sigma_n^2}\Big),
\end{equation}
and since $C_1(x)$ is strictly concave for $x>1$,
\begin{equation}
C_1 \Big(\frac{L}{\sigma_n^2}\Big) > \alpha_1 C_1 \Big(\frac{\gamma_1 L}{\alpha_1 \sigma_n^2}\Big) + \alpha_2 C_1 \Big(\frac{\gamma_2 L}{\alpha_2 \sigma_n^2}\Big),
\end{equation}
according to Jensen's inequality, if $\frac{L}{\sigma_n^2}, \frac{\gamma_1 L}{\alpha_1 \sigma_n^2},\frac{\gamma_2 L}{\alpha_2 \sigma_n^2}>1$.
Theorem \ref{theorem:ergodic_inequality_oma} is proved.



%

%

\begin{IEEEbiography}[{\includegraphics[width=1in,height=1.25in,clip,keepaspectratio]{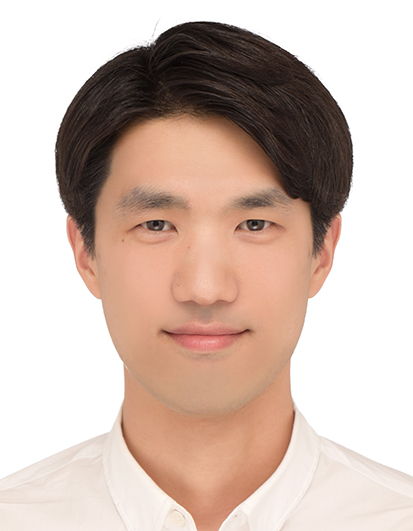}}]{Minseok Choi}
	received the B.S. and M.S. degree in electrical engineering from the Korea Advanced Institute of Science and Technology (KAIST), Daejeon, Republic of Korea, in 2011 and 2013, respectively. He is currently pursuing the Ph.D degree in KAIST. His research interests include wireless caching network, NOMA, 5G Communications, and stochastic network optimization.\end{IEEEbiography}
\begin{IEEEbiography}[{\includegraphics[width=1in,height=1.25in,clip,keepaspectratio]{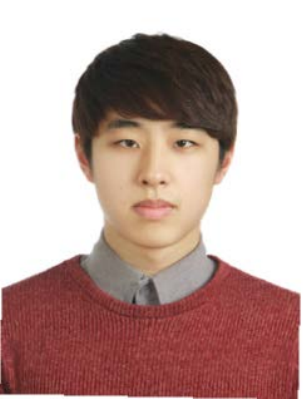}}]{Dong-Jun Han}
	received the B.S. degrees in Mathematics and Electrical Engineering from the Korea Advanced Institute of Science and Technology (KAIST), Daejeon, South Korea, in 2016, where he is currently pursuing the M.S. degree. His research interests include 5G wireless communications and machine learning.\end{IEEEbiography}
\begin{IEEEbiography}[{\includegraphics[width=1in,height=1.25in,clip,keepaspectratio]{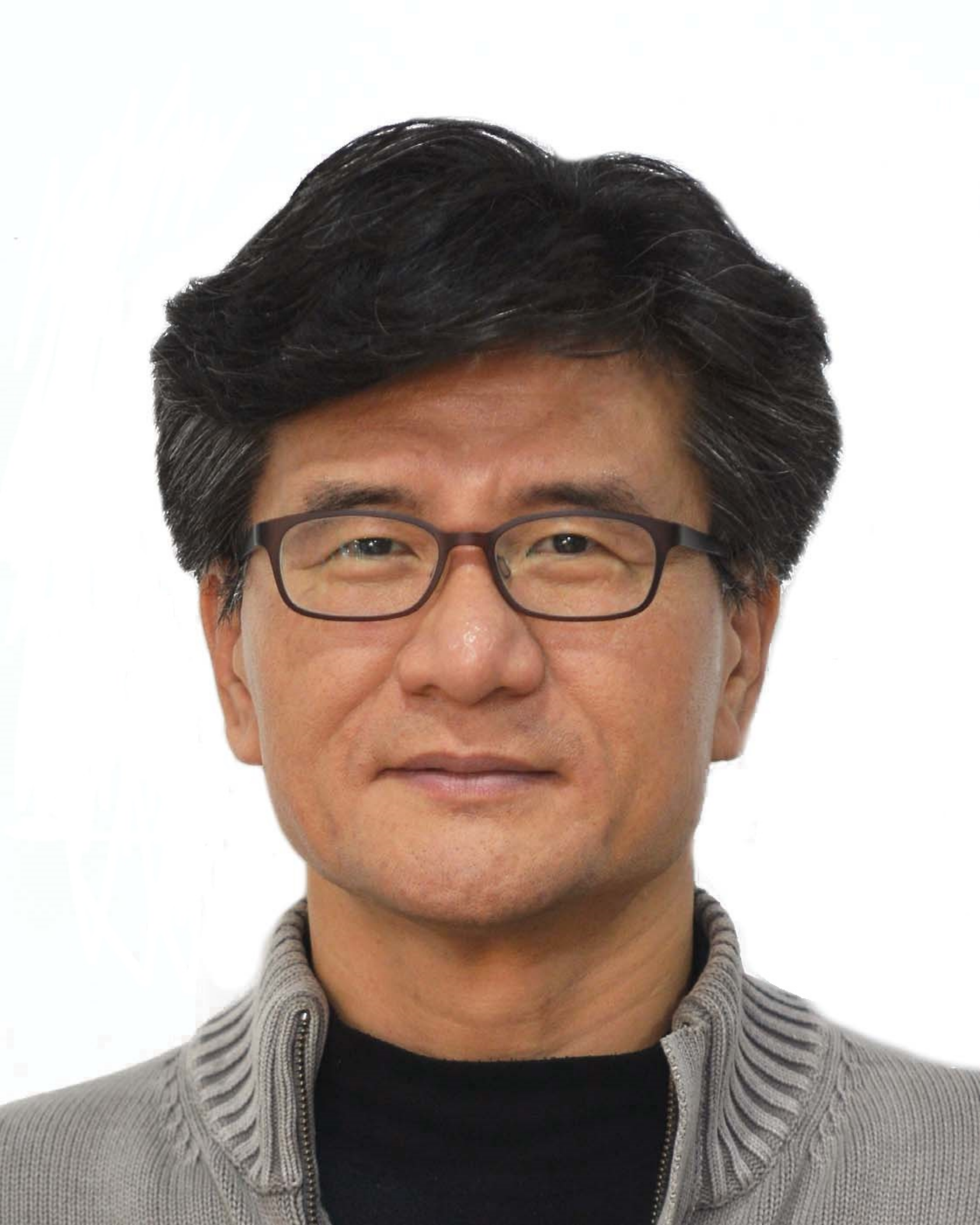}}]{Jaekyun Moon}
	received the Ph.D degree in electrical and computer engineering at Carnegie Mellon University, Pittsburgh, Pa, USA. He is currently a Professor of electrical enegineering at KAIST. From 1990 through early 2009, he was with the faculty of the Department of Electrical and Computer Engineering at the University of Minnesota, Twin Cities. He consulted as Chief Scientist for DSPG, Inc. from 2004 to 2007. He also worked as Chief Technology Officier at Link-A-Media Devices Corporation. His research interests are in the area of channel characterization, signal processing and coding for data storage and digital communication. Prof. Moon received the McKnight Land-Grant Professorship from the University of Minnesota. He received the IBM Faculty Development Awards as well as the IBM Partnership Awards. He was awarded the National Storage Industry Consortium (NSIC) Technical Achievement Award for the invention of the maximum transition run (MTR) code, a widely used error-control/modulation code in commercial storage systems. He served as Program Chair for the 1997 IEEE Magnetic Recording Conference. He is also Past Chair of the Signal Processing for Storage Technical Committee of the IEEE Communications Society.
	He served as a guest editor for the 2001 IEEE JSAC issue on Signal Processing for High Density Recording. He also served as an Editor for IEEE TRANSACTIONS ON MAGNETICS in the area of signal processing and coding for 2001-2006. He is an IEEE Fellow.
\end{IEEEbiography}\vfill




\end{document}